\documentclass{lmcs}
\pdfoutput=1

\usepackage{lastpage}
\lmcsdoi{18}{3}{29}
\lmcsheading{}{\pageref{LastPage}}{}{}%
{Jul.~28,~2021}{Sep.~07,~2022}{}


\usepackage{hyperref}
\usepackage[norelsize,ruled,vlined]{algorithm2e}
\usepackage[utf8]{inputenc}
\usepackage{color}
\usepackage{graphicx}
\usepackage{subcaption}

\newcommand{\Z}{\mathbb Z}
\newcommand{\N}{\mathbb N}

\newcommand{\U}{\mathcal U}

\newcommand{\set}[1]{\left\{ #1 \right\}}

\newcommand{\game}{\mathcal{G}}
\newcommand{\auto}{\mathcal{A}}
\newcommand{\Graph}{\texttt{G}}

\newcommand{\VE}{\text{V}_{\text{Eve}}}
\newcommand{\VA}{\text{V}_{\text{Adam}}}
\newcommand{\WE}{\text{W}_{\text{Eve}}}

\newcommand{\last}{\text{last}}
\newcommand{\dist}{\text{dist}}
\newcommand{\col}{\text{col}}

\newcommand{\F}{\mathcal{F}}

\newcommand{\MP}[1]{\mathtt{MeanPayoff}(#1)}

\newcommand{\proj}{\text{Proj}}

\newcommand{\indeg}{\text{in-deg}}
\newcommand{\outdeg}{\text{out-deg}}

\newcommand{\chain}{\triangleright}
\newcommand{\mchain}{\ \tilde \chain \ }

\newcommand{\Path}{\text{Path}}
\newcommand{\FPath}{\text{Path}_{\text{fin}}}
\newcommand{\IPath}{\text{Path}_{\infty}}

\newcommand{\Det}{\text{Det}}

\newcommand{\Safe}{\mathtt{Safe}}
\newcommand{\Parity}{\mathtt{Parity}}

\newcommand{\len}{\text{len}}

\newcommand{\Inv}{\mathrm{Inv}}

\newcommand{\NP}{\text{NP}}
\newcommand{\coNP}{\text{coNP}}

\newcommand{\cnt}{\texttt{Count}}

\newcommand{\lex}{\mathrm{lex}}
\newcommand{\occ}{\mathrm{occ}}

\newcommand{\tleft}{\mathrm{left}}
\newcommand{\tmiddle}{\mathrm{middle}}
\newcommand{\tright}{\mathrm{right}}

\usepackage{tikz}
\usetikzlibrary{arrows}
\usetikzlibrary{automata}
\usetikzlibrary{shapes,snakes}
\usetikzlibrary{calc}
\usetikzlibrary{patterns}
\usetikzlibrary{shapes.geometric}
\usetikzlibrary{positioning}
\tikzstyle{every node}=[font=\small]
\tikzstyle{eve}=[circle,minimum size=.3cm,draw=gray!90,inner sep=1pt,fill=gray!20,very thick]
\tikzstyle{adam}=[rounded corners=.5,regular polygon,regular polygon
  sides=4,minimum size=.4cm,draw=gray!90,inner sep=1pt,fill=gray!20,very thick]
\tikzstyle{every edge}=[draw,>=stealth',shorten >=1pt]
\tikzstyle{win}=[fill=green!50,draw=green!70!black]
\tikzstyle{lose}=[fill=white,draw=red!70!black]

\tikzstyle{state}=[draw,circle,minimum size=5mm]
\tikzstyle{accepting}=[double]

\begin{document}

\title[The Theory of Universal Graphs for Infinite Duration Games]{The Theory of Universal Graphs\texorpdfstring{\\}{} for Infinite Duration Games}
\titlecomment{A preliminary version of this work was published in the accompanying paper for the invited talk given by Thomas Colcombet at the International Conference on Foundations of Software Science and Computation Structures (FoSSaCS)~\cite{CF19}, see also the corresponding technical reports~\cite{CF18} and~\cite{Fij18}.
The applications to mean payoff games (Section~\ref{sec:mean_payoff}) was published in the International Symposium on Mathematical Foundations of Computer Science (MFCS)~\cite{FGO18}.
The applications to disjunctions of mean payoff and parity games (Section~\ref{sec:disj_mean_payoff_parity})
and disjunctions of mean payoff games (Section~\ref{sec:disj_mean_payoff}) are unpublished.}

\author[T. Colcombet]{Thomas Colcombet\lmcsorcid{0000-0001-6529-6963}\rsuper{a}}
\address{CNRS, IRIF, Universit\'e de Paris, Paris, France}

\author[N. Fijalkow]{Nathana{\"e}l Fijalkow\lmcsorcid{0000-0002-6576-4680}\rsuper{b}}
\address{CNRS, LaBRI, Bordeaux, France, and
The Alan Turing Institute of data science, London, United Kingdom}

\author[P. Gawrychowski]{Pawe{\l} Gawrychowski\lmcsorcid{0000-0002-6993-5440}\rsuper{c}}
\address{University of Wroc{\l}aw, Wroc{\l}aw, Poland}

\author[P. Ohlmann]{Pierre Ohlmann\lmcsorcid{0000-0002-4685-5253}\rsuper{d}}
\address{IRIF, Universit\'e de Paris, Paris, France}

\begin{abstract}
We introduce the notion of universal graphs as a tool for constructing algorithms solving games of infinite duration such as parity games and mean payoff games.
In the first part we develop the theory of universal graphs, with two goals: showing an equivalence and normalisation result between different recently introduced related models, and constructing generic value iteration algorithms for any positionally determined objective.
In the second part we give four applications: to parity games, to mean payoff games, to a disjunction between a parity and a mean payoff objective, and to disjunctions of several mean payoff objectives. For each of these four cases we construct algorithms achieving or improving over the best known time and space complexity.
\end{abstract}

\maketitle

\section{Introduction}
\label{sec:intro}
\textbf{Games of infinite duration} are a widely studied model in several fields of Computer Science including Program Verification, Model checking, Automata Theory, Logic, Finite Model Theory, and Database Theory: 
the interactions between the players can model a range of (non-terminating) systems and are therefore used for analysing, verifying, and synthesising them.
There is a wide range of game models: in this paper we consider two player zero sum deterministic (as opposed to stochastic) games with perfect information.
Two of the most important open problems in this field concern parity games and mean payoff games:
in both cases the complexity of solving them is in $\NP$ and in $\coNP$, but not known to be in polynomial time.
This complexity status was once shared with primality testing, linear programming, and other famous problems, which all turned out to be solvable in polynomial time.
Yet although both problems have attracted a lot of attention over the past three decades, they remain widely open and exciting research directions.
This paper introduces the combinatorial notion of universal graphs and constructs conceptually simple algorithms achieving state of the art time and space complexity for both parity and mean payoff games, and beyond.
In order to motivate the introduction of universal graphs, let us revisit the recent developments on parity games.

\vskip1em 
\textbf{Parity games} are a central model in the study of logic and automata over infinite trees and for their tight relationship with model-checking games for the modal $\mu$-calculus.
The breakthrough result of Calude, Jain, Khoussainov, Li, and Stephan~\cite{CJKLS17} was to construct a quasipolynomial time algorithm for solving parity games. 
Following decades of exponential and subexponential algorithms, this very surprising result triggered further research: soon after further quasipolynomial time algorithms were constructed reporting almost the same complexity.
Let us classify them in two families; we will motivate this classification in the next paragraphs.
\begin{itemize}
	\item The first family of algorithms includes the original algorithm by Calude, Jain, Khoussainov, Li, and Stephan~\cite{CJKLS17} (see also~\cite{FJSSW17} for a presentation of the algorithm as value iteration), then the succinct progress measure lifting algorithm by Jurdzi{\'n}ski and Lazi{\'c}~\cite{JL17}, and the register games algorithm by Lehtinen~\cite{Leh18} (see also the journal version~\cite{LehtinenB20}).

	\item The second family of algorithms is so-called McNaughton Zielonka-type algorithms, inspired by the exponential time algorithm by Zielonka~\cite{Zielonka98} specialising an algorithm due to McNaughton~\cite{McNaughton:1993}. 
The first quasipolynomial time algorithm is due to Parys~\cite{Parys19}. Its complexity was later improved by Lehtinen, Schewe, and Wojtczak~\cite{LSW19}.
Recently Jurdzi{\'n}ski and Morvan~\cite{JM20} constructed a universal attractor decomposition algorithm encompassing all three algorithms, which was then refined and further studied in~\cite{JMOT20}.
\end{itemize}

\vskip1em 
\textbf{Separating automata} have been introduced by Boja{\'n}czyk and Czerwi{\'n}ski~\cite{BC18} in an effort to understand the first quasipolynomial time algorithm: they have extracted from the algorithm the construction of a separating automaton and presented the original algorithm as solving a safety game obtained as the product of the original game by the separating automaton.
Later Czerwi{\'n}ski, Daviaud, Fijalkow, Jurdzi{\'n}ski, Lazi{\'c}, and Parys~\cite{CDFJLP18} showed that the other two algorithms also induce separating automata.
Separating automata are deterministic; we refer to Subsection~\ref{subsec:complexity_parity} for a more in-depth discussion on the use of non-determinism.

\vskip1em 
\textbf{Universal trees} have been defined by Fijalkow~\cite{Fij18} (superseded by~\cite{CDFJLP18}) to rephrase the second quasipolynomial time algorithm~\cite{JL17}, constructing a family of value iteration algorithms parametrised by the choice of a universal tree. Within this family the algorithm~\cite{JL17} is optimal: the universal tree (implicitly) constructed in~\cite{JL17} matches (up to a polynomial factor) the lower bounds offered in~\cite{Fij18}.
The main contribution of~\cite{CDFJLP18} is to show that any separating automaton contains a universal tree in its set of states;
in other words, the three algorithms in the first family induce three (different) constructions of universal trees.

The generic algorithm constructed in~\cite{JM20} is parametrised by the choice of two universal trees (one of each player),
and it is argued that choosing appropriate pairs of universal trees yields the algorithms from~\cite{Zielonka98,Parys19,LSW19}.

\vskip1em
\textbf{Ubiquity of universal trees.}
In summary, all quasipolynomial time algorithms for parity games (and some exponential ones) fall in one of two families, and both are based on the combinatorial notion of universal trees.
This has a number of consequences:
\begin{itemize}
	\item \textit{simplicity}: all algorithms are instances of one of two generic algorithms,
	\item \textit{combinatorial abstraction}: algorithms have been reduced to a purely combinatorial notion on trees which is easier to study and captures the combinatorial contents and complexity of the algorithms,
	\item \textit{lower bounds}: the lower bounds on universal trees imply a complexity barrier applying to both families of algorithms, hence to all known quasipolynomial time algorithms.
\end{itemize}

\vskip1em
\textbf{Beyond parity games.}
There are a number of other objectives of interest, the most prominent one being mean payoff.
We introduce universal graphs for any objective which is positionally determined for at least one of the two players
(with some additional technical assumptions, see Subsection~\ref{subsec:assumptions}) 
with the following goals: constructing generic algorithms, 
reducing the construction of algorithms to the combinatorial understanding of universal graphs,
and offering lower bounds to argue about the optimality of the algorithms within this framework.
As we shall see, when the objective corresponds to parity games, universal graphs instantiate as universal trees.

\vskip1em
\textbf{Our contributions and outline.}
The first part of the paper is to develop the theory of universal graphs.

After the mandatory preliminaries in Section~\ref{sec:preliminaries},
we start in Section~\ref{sec:reductions} by constructing a generic algorithm for solving games by reduction to safety games.
We present three sufficient conditions for this algorithm to be correct:
using a separating automaton, using a good-for-small-games automaton, or using a universal graph.
We give in Section~\ref{sec:saturation} the main technical result of the theory of universal graphs: 
it is a normalisation result through a saturation process showing the equivalence of the three models.
A first consequence of this normalisation is an efficient implementation of the generic algorithm
defined in Section~\ref{sec:reductions} as a value iteration algorithm, constructed in Section~\ref{sec:value_iteration}.

\vskip1em 
The second part of the paper is to give four applications of the universal graph approach: 
for four classes of objectives, we construct universal graphs and derive algorithms for solving the corresponding games. 
Since the complexity of the algorithm is proportional to the size of the universal graph, our goal is to construct as small as possible universal graphs.
In all four cases our algorithms achieve the best known time and space complexity for solving the corresponding games, even significantly improving the state of the art in one case.
Along the way we will explain that many existing algorithms are actually instances of our framework, 
meaning are solving implicitly or explicitly a product game for some universal graph.

In these four applications we use two different approaches for studying universal graphs:
\begin{itemize}
	\item Both for parity and mean payoff objectives, we first study the properties of saturated universal graphs,
	find a combinatorial characterisation: using trees for parity objectives (Section~\ref{sec:parity})
	and subsets of the integers for mean payoff objectives (Section~\ref{sec:mean_payoff}).
	Using these combinatorial characterisations we prove upper and lower bounds on the size of universal graphs: 
	the upper bounds induce algorithms for solving the corresponding games,
	and the lower bounds imply that the algorithms we construct are optimal within this class of algorithms.
	
	\item Both for disjunctions of a parity and a mean payoff objective, and for disjunctions of several mean payoff objectives\footnote{Although this distinction is irrelevant for mean payoff games, it matters whether one uses the $\limsup$ or $\liminf$ semantics when considering their disjunctions (see~\cite{VelnerC0HRR15}). Here, we only consider $\liminf$ semantics, which is algorithmically tractable.},
	we investigate how to combine existing universal graphs for each objective into universal graphs for their disjunctions.
	We provide some general results towards this goal, in particular a reduction to strongly connected graphs 
	in Section~\ref{sec:disj_mean_payoff}.
	For disjunctions of a parity and a mean payoff objective (Section~\ref{sec:disj_mean_payoff_parity}) 
	it is more convenient to reason with separating automata, 
	which thanks to our general equivalence theorem (Theorem~\ref{thm:main})
	is equivalent with respect to size, our main concern for algorithmic applications.	
\end{itemize}

\section{Preliminaries}
\label{sec:preliminaries}
We write $[i,j]$ for the interval $\set{i,i+1,\dots,j-1,j}$, and use parentheses to exclude extremal values,
so for instance $[i,j)$ is $\set{i,i+1,\dots,j-1}$.
We let $C$ denote a set of colours and write $C^*$ for finite sequences of colours (also called finite words),
$C^+$ for finite non-empty sequences, and $C^\omega$ for infinite sequences (also called infinite words).
The empty word is $\varepsilon$.

\subsection{Graphs are safety automata and safety automata are graphs}~\\[-1.5em]

\paragraph*{Graphs}
We consider edge labelled directed graphs: 
a graph $G$ is given by a (finite) set $V$ of vertices and a (finite) set $E \subseteq Q \times C \times Q$ of edges,
with $C$ a set of colours, so we write $G = (V,E)$.
An edge $(v,c,v')$ is from the vertex $v$ to the vertex $v'$ and is labelled by the colour $c$.
The size of a graph is its number of vertices.
A vertex $v$ for which there exists no outgoing edges $(v,c,v') \in E$ is called a sink.

\paragraph*{Paths}
A path $\pi$ is a (finite or infinite) sequence 
\[
\pi = v_0 c_0 v_1 c_1 v_2 \dots
\]
where for all $i$ we have $(v_i,c_i,v_{i+1}) \in E$.
If it is finite, that is, $\pi= v_0 c_0 \dots c_{i-1} v_i$, we call $\len(\pi) = i \geq 0$ the length of $\pi$, 
and use $\last(\pi)$ to denote the last vertex $v_i$.
The length of an infinite path is $\infty$.
We say that $\pi$ starts from $v_0$ or is a path from $v_0$, 
and in the case where $\pi$ is finite we say that $\pi$ is a path ending in $\last(\pi)$ or simply a path to $\last(\pi)$.
We say that $v'$ is reachable from $v$ if there exists a path from $v$ to $v'$.
We let $\pi_{\le i}$ denote the prefix of $\pi$ ending in $v_i$, meaning 
$\pi_{\le i} = v_0 c_0 v_1 \dots c_{i-1} v_i$.
A cycle is a path from a vertex to itself of length at least one.

We use $\FPath(G), \IPath(G)$ and $\Path(G)$ to denote respectively the sets of finite paths of $G$, infinite paths of $G$, and their union. We sometimes drop $G$ when it is clear from context. 
For $v_0 \in V$ we also use $\FPath(G,v_0), \IPath(G,v_0)$ and $\Path(G,v_0)$ 
to refer to the sets of paths starting from $v_0$. 
We use $\col(\pi)$ to denote the (finite or infinite) sequence of colours $c_0 c_1 \dots$ induced by $\pi$.

\paragraph*{Objectives}
An objective is a set $\Omega \subseteq C^\omega$ of infinite sequences of colours.
We say that a sequence of colours belonging to $\Omega$ satisfies $\Omega$, and extend this terminology to infinite paths $\pi$ with $\col(\pi) \in \Omega$.

\begin{defi}[Graphs satisfying an objective]
Let $\Omega$ be an objective and $G$ a graph.
We say that $G$ satisfies $\Omega$ if all infinite paths in $G$ satisfy $\Omega$.
\end{defi}

\paragraph*{Safety automata}
A safety automaton over the alphabet $C$ is given by a finite set of states $Q$ and a transition function 
$\Delta \subseteq Q \times C \times Q$, so we write $\auto = (Q,\Delta)$.
Since we are only considering safety automata in this paper, we omit the adjective safety and simply speak of an automaton.
A first very important remark is that a (safety) automaton is simply a graph; 
this point of view will be very important and fruitful in this paper.

An automaton $\auto$ recognises the language $L(\auto)$ of infinite words defined as follows:
\[
L(\auto) = \set{\col(\pi) : \pi \text{ is an infinite path in } \auto}.
\]
The class of deterministic automata has two additional properties: 
the automaton $\auto$ has a designated initial state $q_0$, and for any state $q \in Q$ and colour $c \in C$ 
there is at most one transition of the form $(q,c,q')$ in $\Delta$.
In that case we define $\delta : Q \times C \to Q$ such that 
$\delta(q,c)$ is the unique $q'$ where $(q,c,q') \in \Delta$ if it is defined. 
We extend $\delta$ to sequences of colours by the formulas 
$\delta^*(q,\epsilon) = q$ and $\delta^*(q,wc) = \delta(\delta^*(q,w),c)$.
The language recognised by a deterministic automaton $\auto$ is $L(\auto)$ defined by
$L(\auto) = \set{\col(\pi) : \pi \text{ is an infinite path from $q_0$}}$.
Note that if $w$ is a finite prefix of a word in $L(\auto)$, then $\delta^*(q_0,w)$ is well defined.

Following our convention to identify non-deterministic automata and graphs, we say that $\auto$ satisfies $\Omega$ 
if $L(\auto) \subseteq \Omega$.

\subsection{Games}~\\[-1.5em]

\paragraph*{Arenas}
An arena is given by a graph $G=(V,E)$ together with a partition $V = \VE \uplus \VA$ of its set of vertices describing which player controls each vertex.
We represent vertices controlled by Eve with circles and those controlled by Adam with squares.

\paragraph*{Games}
A game is given by an arena and an objective $\Omega$.
We often let $\game$ denote a game, and use $G=(V,E)$ for the underlying graph.
The size of $\game$ is defined to be the size of the underlying graph $G$.
It is played as follows.
A token is initially placed on some vertex $v_0$, and the player who controls this vertex pushes the token along an edge, 
reaching a new vertex; the player who controls this new vertex takes over, 
and this interaction goes on, either forever and describing an infinite path, or until reaching a sink.

We say that a path is winning\footnote{We always take the point of view of Eve, so winning means winning for Eve, and similarly a strategy is a strategy for Eve.} if it is infinite and satisfies $\Omega$, or finite and ends in a sink controlled by Adam.
The definition of a winning path includes the following usual convention: if a player cannot move they lose,
in other words sinks controlled by Adam are winning (for Eve) and sinks controlled by Eve are losing.

We extend the notations $\FPath, \IPath$ and $\Path$ to games by considering the underlying graph.

\paragraph*{Strategies}
A strategy in $\game$ is a partial map $\sigma : \FPath(\game) \to E$ such that $\sigma(\pi)$ is 
an outgoing edge of $\last(\pi)$ when it is defined.
We say that a path $\pi = v_0 c_0 v_1 \dots$ is consistent with $\sigma$ if
for all $i < \len(\pi)$, if $v_i \in \VE$ then $\sigma$ is defined over $\pi_{\leq i}$ and 
$\sigma(\pi_{\le i}) = (v_i, c_i, v_{i+1})$.
A consistent path with $\sigma$ is maximal if it is not the strict prefix of a consistent path with $\sigma$
(in particular infinite consistent paths are maximal).

A strategy $\sigma$ is winning from $v_0$ if all maximal paths consistent with $\sigma$ are winning.
Note that in particular, if a finite path $\pi$ is consistent with a winning strategy $\sigma$ and ends in a vertex which belongs to Eve, then $\sigma$ is defined over $\pi$.
We say that $v_0$ is a winning vertex of $\game$ or that Eve wins from $v_0$ in $\game$ 
if there exists a winning strategy from $v_0$, and let $\WE(\game)$ denote the set of winning vertices.

\paragraph*{Positional strategies}
Positional strategies make decisions only considering the current vertex. 
Such a strategy is given by $\widehat \sigma : \VE \to E$.
A positional strategy induces a strategy $\sigma : \FPath \to E$ from any vertex $v_0$ 
by setting ${\sigma}(\pi) = \widehat{\sigma}(\last(\pi))$ when $\last(\pi) \in \VE$.

\begin{defi}[Positionally determined objectives]
We say that an objective $\Omega$ is \textit{positionally determined} if for every game with objective $\Omega$ and vertex $v_0$,
if $v_0 \in \WE(\game)$ then there exists a positional winning strategy from $v_0$.
\end{defi}

Given a game $\game$, a vertex $v_0$, and a positional strategy $\sigma$ we let $\game[\sigma,v_0]$ denote the graph obtained by restricting $\game$ to vertices reachable from $v_0$ by playing $\sigma$ and to the moves prescribed by $\sigma$.
Formally, the set of vertices and edges is
\[
\begin{array}{lll}
V[\sigma,v_0] & = & \set{v \in V : \text{there exists a path from } v_0 \text{ to } v \text{ consistent with } \sigma}, \\
E[\sigma,v_0] & = & \set{(v,c,v') \in E : v \in \VA \text{ or } \left( v \in \VE \text{ and }\right.
\left. \sigma(v) = (v,c,v') \right)} \\
& \cap & V[\sigma,v_0] \times C \times V[\sigma,v_0].
\end{array}
\]

\begin{defi}[Prefix independent objectives]
We say that an objective $\Omega$ is prefix independent if for all $u \in C^*$ and $v \in C^\omega$, 
we have $uv \in \Omega$ if and only if $v \in \Omega$.
\end{defi}

\begin{fact}\label{fact:game_mp}
Let $\Omega$ be a prefix independent objective, $\game$ a game, $v_0$ a vertex, and $\sigma$ a positional strategy.
Then $\sigma$ is winning from $v_0$ if and only if the graph $\game[\sigma,v_0]$ satisfies $\Omega$
and does not contain any sink controlled by Eve.
\end{fact}

\paragraph*{Safety games}
The safety objective $\Safe$ is defined over the set of colours $C = \set{\varepsilon}$ by $\Safe = \set{\varepsilon^\omega}$: 
in words, all infinite paths are winning, so losing for Eve can only result from reaching a sink that she controls.
Since there is a unique colour, when manipulating safety games we ignore the colour component for edges.
Note that in safety games, strategies can equivalently be defined as maps to $V$ (and not $E$):
only the target of an edge matters when the source is fixed, since there is a unique colour.
We sometimes use this abuse of notations for the sake of simplicity and conciseness.

The following result is folklore, we refer to Section~\ref{sec:value_iteration} for more details on the algorithm.
The statement uses the classical RAM model with word size $w = \log(|V|)$.

\begin{thm}\label{thm:safety_games}
There exists an algorithm computing the set of winning vertices of a safety game running in time $O(|E|)$.
\end{thm}

\subsection{Assumptions on objectives}
\label{subsec:assumptions}

\begin{defi}[Neutral letter]
Let $\Omega \subseteq C^\omega$ an objective and $0 \in C$.
For $w \in C^\omega$ we let $\proj(w)$ denote the projection of $w$ on $C \setminus \set{0}$, that is, the finite or infinite word obtained by removing all occurrences of $0$ in $w$.
We say that $0$ is a \textit{neutral letter} if
for all $w \in C^\omega$, if $\proj(w)$ is finite then $w \in \Omega$, and if $\proj(w)$ is infinite,
then $w \in \Omega$ if and only if $\proj(w) \in \Omega$.
\end{defi}

In this paper, we will prove results applying to general objectives $\Omega$ satisfying three assumptions:
\begin{itemize}
	\item $\Omega$ is positionally determined,
	\item $\Omega$ is prefix independent,
	\item $\Omega$ has a neutral letter.
\end{itemize}
The first assumption is at the heart of the approach; the other two are here rather for technical convenience. In fact, it is not known whether there are positional objectives for which artificially adding a neutral letter breaks positionality~\cite{Kopczynski06}. Regarding prefix independence, we believe that our techniques can be adapted by fixing an initial vertex (both in the games under consideration and in the universal graph), but chose not to pursue this path for the sake of simplicity and convenience.

\section{Reductions to safety games}
\label{sec:reductions}
This section describes different ways to reduce solving a game $\game$ with a positionally determined objective $\Omega$ 
to solving a (larger) equivalent safety game, using a product construction with an automaton (or equivalently, a graph).
In this section we introduce the concepts needed to prove the following theorem.

\begin{thm}\label{thm:reductions_to_safety}
Let $n \in \N$, $\Omega$ a positionally determined objective, $\game$ a game of size $n$ with objective $\Omega$, 
and $\auto = (Q, \Delta)$ be either an $(n, \Omega)$-separating automaton (defined in Subsection~\ref{subsec:separating}) an $(n, \Omega)$-GFSG automaton (defined in Subsection~\ref{subsec:GFSG}) or an $(n, \Omega)$-universal graph (defined in Subsection~\ref{subsec:universal}). 
Then the winning region in $\game$ can be computed\footnote{In the RAM model with word size $w = \log(|V|) + \log(|Q|)$.} 
in time $O(|E| \cdot |\Delta|)$.
\end{thm}

In all three cases, this is obtained by reducing $\game$ to the safety game $\game \chain \auto$, and then applying Theorem~\ref{thm:safety_games}.
The three notions (separating automata, GFSG automata, and universal graphs) are sufficient conditions for the two games 
$\game$ and $\game \chain \auto$ to be equivalent.
We start by describing the product construction $\game \chain \auto$.

\subsection{Product construction}
\label{subsec:product}

Let $\auto = (Q, \Delta)$ be an automaton and $\game$ a game with objective $\Omega$.
We define the chained game $\game \chain \auto$ as the safety game with vertices $V' = \VE' \uplus \VA'$ and edges $E'$ given by
\[
\begin{array}{lll}
\VE' & = & \VE \times Q \uplus E\times Q, \\
\VA' & = & \VA \times Q, \\
E' & = & \{((v,q), ((v,c,v'),q)) : q \in Q, (v,c,v') \in E\} \\
& \cup & \{(((v,c,v'),q), (v',q')) : (v,c,v') \in E, (q,c,q') \in \Delta\}.
\end{array}
\]

In words, from $(v,q) \in V \times Q$, the player whom $v$ belongs to chooses an edge $e=(v,c,v') \in E$, and the game progresses to $(e,q)$.
Then Eve chooses a transition $(q,c,q')$ in $\auto$, and the game continues from $(v',q')$.
Intuitively, we simulate playing in $\game$, but additionally Eve has to simultaneously  follow a path in $\auto$ which reads the same colours.
Note that the obtained game $\game \chain \auto$ is a safety game: Eve wins if she can play forever or end up in a sink controlled by Adam.

\begin{lem}\label{lem:product_harder}
Let $\Omega$ be an objective, $\game$ a game with objective $\Omega$, $\auto$ an automaton, and $v_0 \in V$.
If $\auto$ satisfies $\Omega$ and $q_0 \in Q$ is such that Eve wins from $(v_0,q_0)$ in $\game \chain \auto$, then Eve wins from $v_0$ in $\game$.
\end{lem}

Intuitively, if Eve can play forever in $\game \chain \auto$, and all paths in $\auto$ are labelled by sequences which belong to $\Omega$, then Eve can use the same strategy in $\game$ and obtain a path satisfying $\Omega$.
The formal proof requires some technical care for transferring the strategy from $\game \chain \auto$ to $\game$.

\begin{proof}
Let $\sigma'$ be a winning strategy from $(v_0, q_0)$ in $\game \chain \auto$.
We construct a strategy $\sigma$, from $v_0$ in $\game$ which simulates $\sigma'$ in the following sense,
which we call the simulation property:
\begin{center}
for any path $\pi = v_0 c_0 \dots c_{i-1} v_i$ in $\game$ which is consistent with $\sigma$, \\
there exists a sequence of states $q_0 \dots q_i$ such that $\pi' = (v_0, q_0)(e_0,q_0) \dots (v_i, q_i)$, \\
with $e_j = (v_j,c_j,v_{j+1})$ for $j < i$, is a path in $\game \chain \auto$ consistent with $\sigma'$.
\end{center}

We define $\sigma$ over finite paths $v_0c_0 \dots c_{i-1}v_i$ with $v_i \in \VE$ by induction over $i$, so that the simulation property holds.
For $i=0$ there is nothing to prove since $\pi'=(v_0,q_0)$ is a path in $\game \chain \auto$ consistent with $\sigma'$.
Let $\pi = v_0 c_0 \dots c_{i-1} v_{i}$ be a path in $\game$ consistent with $\sigma$ and $v_i \in \VE$, 
we want to define $\sigma(\pi)$.
Thanks to the simulation property there exists $q_0 \dots q_i$ such that 
$\pi' = (v_0, q_0)(e_0, q_0) \dots (v_{i}, q_{i})$ is a path in $\game \chain \auto$ consistent with $\sigma'$.
Then $(v_i,q_i) \in \VE'$, so since it is winning, $\sigma'$ is defined over $\pi'$, let us write $\sigma'(\pi') = (e_i,q_i)$ with $e_i = (v_i,c_i,v_{i+1})$. 
We set $\sigma(\pi) = e_i$.

To conclude the definition of $\sigma$ we need to show that the simulation property extends to paths
of length $i+1$.
Let $\pi_{i+1} = \pi\ c_i v_{i+1} = v_0 c_0 \dots v_i c_i v_{i+1}$ be consistent with $\sigma$.
We apply the simulation property to $\pi$ to construct $q_0 \dots q_i$ such that
$\pi' = (v_0, q_0)(e_0, q_0) \dots (v_{i}, q_{i})$ is a path in $\game \chain \auto$ consistent with $\sigma'$.
Let us consider $\pi'_{+1} = \pi'\ (e_i,q_i)$ with $e_i = (v_i,c_i,v_{i+1})$, we claim that $\pi'_{+1}$ is consistent with $\sigma'$.
Indeed, if $v_i \in \VA$, there is nothing to prove, and if $v_i \in \VE$, this holds by construction: 
since $\sigma(\pi) = e_i$ we have $\sigma'(\pi') = (e_i,q_i)$.
Now $(e_i,q_i) \in \VE$ so $\sigma'$ is defined over $\pi'_{+1}$, and we let $\sigma'(\pi'_{+1}) = (v_{i+1}, q_{i+1})$.
Then $\pi'_{+2} = \pi'_{+1} (v_{i+1}, q_{i+1})$ is consistent with $\sigma'$.
This concludes the inductive proof of the simulation property together with the definition of $\sigma$.

\vskip1em
We now prove that $\sigma$ is a winning strategy from $v_0$. 
Let $\pi=v_0 c_0 v_1 \dots$ be a maximal consistent path with $\sigma$, and let $\pi' = (v_0,q_0)(e_0,q_0)(v_1,q_1)(e_1,q_1) \dots$ be the corresponding path in $\game \chain \auto$ consistent with $\sigma'$.
Let us first assume that $\pi$ is finite, and let $v_i = \last(\pi)$. If $v_i \in \VE$, then by construction $\sigma$ is defined over $\pi$, so the path $v_0c_0 \dots v_i c_i v_{i+1}$, where $\sigma(\pi)=(v_i,c_i,v_{i+1})$ is consistent with $\sigma$, contradicting maximality of $\pi$. If however $v_i \in \VA$, then by maximality of $\sigma$, $v_i$ must be a sink, hence $\pi$ is winning.
Now if $\pi$ is infinite then so is $\pi'$, and then $q_0 c_0 q_1 c_1 \dots$ is a path in $\auto$, so $\col(\pi) = c_0 c_1 \dots \in L(\auto) \subseteq \Omega$, and $\pi$ is winning.
We conclude that $\sigma$ is a winning strategy from $v_0$ in $\game$.
\end{proof}

Hence, assuming that $\auto$ satisfies $\Omega$, winning $\game$ is always easier for Eve than winning $\game \chain \auto$. The three following sections introduce three different settings where the converse holds, that is, $\game$ and $\game \chain \auto$ are equivalent. In particular, via Theorem~\ref{thm:safety_games}, this leads to algorithms with runtime $O(|E'|)$ 
which is bounded by $O(|E| \cdot |\Delta|)$ for computing the winning region in $\game$.
%

\subsection{Separating automata}
\label{subsec:separating}
It is not hard to see that if $\auto$ is deterministic and recognises exactly $\Omega$, then $\game$ and $\game \chain \auto$ are equivalent.
However, for many objectives $\Omega$ (for instance, parity or mean payoff objectives), a simple topological argument shows that such deterministic automata do not exist.
Separating automata are defined by introducing $n$ as a parameter, and relaxing the condition $L(\auto)=\Omega$ to the weaker condition $\Omega^{\mid n} \subseteq L(\auto) \subseteq \Omega$, where $\Omega^{\mid n}$ is the set of infinite sequence of colours that label paths from graphs of size at most $n$ satisfying $\Omega$.
Formally,
\[
\Omega^{\mid n} = \{\col(\pi) \mid \pi \in \IPath(G), G \text{ has size at most $n$ and satisfies $\Omega$}\}.
\]

\begin{defi}
An $(n, \Omega)$-separating automaton $\auto$ is a deterministic automaton such that
\[
\Omega^{\mid n} \subseteq L(\auto) \subseteq \Omega.
\]
\end{defi}

The definition given here slightly differs from the original one given by Boja{\'n}czyk and Czerwi{\'n}ski~\cite{BC18},
who use a different relaxation $\Omega_n$ satisfying $\Omega^{\mid n} \subseteq \Omega_n \subseteq \Omega$.

\begin{thm}
\label{thm:separating_automata}
Let $\Omega$ be a positionally determined objective, $\auto$ an $(n,\Omega)$-separating automaton, $\game$ a game of size $n$ with objective $\Omega$, and $v_0 \in V$. 
Then Eve wins from $v_0$ in $\game$ if and only if she wins from $(v_0,q_0)$ in $\game \chain \auto$.
\end{thm}

We postpone the proof of Theorem~\ref{thm:separating_automata} to the next subsection which discusses the more general setting of GFSG automata.

\subsection{Good-for-small-games Automata}
\label{subsec:GFSG}
Good-for-small-games (GFSG) automata extend separating automata by introducing some (controlled) non-determinism.
The main motivation for doing so is that non-deterministic automata are more succinct than deterministic ones,
potentially leading to more efficient algorithms. 
The notion of good-for-small-games automata we introduce here is inspired by~\cite{HenzingerPiterman06}, and more precisely the variant called history-deterministic automata used in the context of the theory of regular cost functions (\cite{Colcombet09}, see also~\cite{ColcombetF16}). 
We refer to Subsection~\ref{subsec:complexity_parity} for a discussion on the use of this model for solving parity games.

In a good-for-games automaton non-determinism can be resolved on the fly, only by considering the word read so far.
For GFSG automata, this property is only required for words in $\Omega^{\mid n}$.

\begin{defi}[Good-for-small-games automata]
An automaton $\auto$ is $(n,\Omega)$-GFSG if it satisfies $\Omega$, and moreover, there exists a state $q_0$ and a partial map $\sigma_\auto : C^+ \to Q$ such that for all $c_0 c_1 \dots \in \Omega^{\mid n}$, 
$q_0 c_0 q_1 \dots$ defines an infinite path in $\auto$ with $q_i = \sigma_\auto(c_0 \dots c_{i-1})$ for all $i \geq 1$.
\end{defi}

Note that $\sigma_\auto$ is defined over all prefixes of $\Omega^{\mid n}$.
By a slight abuse, we refer to $\sigma_\auto$ as the strategy of the $(n, \Omega)$-GFSG automaton $\auto$.
This is justified by the fact that $\sigma_\auto$ can formally be seen as a strategy in a game where 
Adam inputs colours and Eve builds a run in $\auto$, and 
Eve wins if whenever the sequence of colours is in $\Omega^{\mid n}$ then the play is infinite.

If $\auto$ is a separating automaton with initial state $q_0$, the partial map $w \mapsto \delta^*(q_0,w)$ is a strategy which indeed makes it GFSG.\@ 
In this sense, the following theorem generalises Theorem~\ref{thm:separating_automata}.

\begin{thm}
\label{thm:GFSG_automata}
Let $\Omega$ be a positionally determined objective, $\auto$ an $(n, \Omega)$-GFSG automaton, $\game$ a game of size $n$ with objective $\Omega$ and $v_0 \in V$. 
Then Eve has a winning strategy in $\game$ from $v_0$
if and only if she has a winning strategy in $\game \chain \auto$ from $(v_0,q_0)$.
\end{thm}

\begin{proof}
The ``if'' follows directly from Lemma~\ref{lem:product_harder}. 
Conversely, we assume that Eve wins $\game$ from $v_0$, and let $\sigma$ be a winning strategy from $v_0$.
Since $\Omega$ is positionally determined, we can choose $\sigma$ to be positional.
We define a strategy $\sigma'$ over finite paths $\pi' = (v_0,q_0)(e_0,q_0) \dots$ with $e_i = (v_i,c_i,v_{i+1})$ for all $i$
and $\last(\pi') \in \VE'$ as follows:
\[
\begin{array}{llll}
\sigma'(\pi') & = & (\sigma(v_i), q_i) & \text{ if } \last(\pi') = (v_i,q_i) \in \VE \times Q, \\
\sigma'(\pi') & = & (v_{i+1}, \sigma_\auto(c_0 \dots c_i))) & \text{ if } \last(\pi') = (e_i,q_i) \in E \times Q.
\end{array}
\]
Intuitively, we simply play as prescribed by $\sigma$ on the first coordinate, and as prescribed by $\sigma_\auto$ on the second.
We first show the following simulation property: 

\begin{center}
Let $\pi' = (v_0,q_0)(e_0,q_0) \dots (v_i,q_i)(e_i,q_i)$ \\
with $e_j = (v_j,c_j,v_{j+1})$ for all $j \le i$ be consistent with $\sigma'$, \\
then $\pi = v_0 c_0 \dots v_i c_{i+1} v_{i+1}$ is consistent with $\sigma$.
\end{center}

We proceed by induction over $i$, treating together the base case $i=0$ and the inductive case.
Let us consider $\pi'$ and $\pi$ as in the simulation property.
In both the base and inductive cases, $\pi_{\le i} = v_0 c_0 \dots c_{i-1} v_i$ is consistent with $\sigma$, either trivially or by inductive hypothesis.
There are two cases: either $v_i \in \VA$, and then since $\pi_{\le i}$ is consistent with $\sigma$ 
it follows that $\pi$ is consistent with $\sigma$,
or $v_i \in \VE$, and then since $\pi'$ is consistent with $\sigma'$ and by the first item of the definition of $\sigma'$ we have 
$\sigma(v_i) = e_i$, again $\pi$ is consistent with $\sigma$.
This concludes the inductive proof of the simulation property.

\vskip1em
We now show that $\sigma'$ is a winning strategy in $\game \chain \auto$ from $(v_0,q_0)$.
Consider $\pi'$ a maximal consistent path in $\game \chain \auto$ with $\sigma'$.
Since $\sigma'$ is defined over consistent paths ending in $\VE$, if, for contradiction, $\pi'$ is not winning in the safety game $\game \chain \auto$, then it is finite and ends in a sink controlled by Eve.
There are two cases: either $\last(\pi') = (v_i,q_i) \in \VE \times Q$, or $\last(\pi') = (e_i,q_i) \in E \times Q$.

In the first case, the simulation property applied to 
$\pi' = (v_0,q_0)(e_0,q_0) \dots (e_{i-1},q_{i-1})$ 
implies that $\pi = v_0 c_0 \dots c_i v_i$ is consistent with $\sigma$.
Since $(v_i,q_i)$ is a sink this implies that $v_i$ is a sink (controlled by Eve), contradicting that $\sigma$ is a winning strategy from $v_0$.

In the second case the simulation property applied to
 $\pi' = (v_0,q_0)(e_0,q_0) \dots (v_i,q_i)(e_i,q_i)$
implies that $\pi = v_0 c_0 \dots v_i c_i v_{i+1}$ is consistent with $\sigma$.
In other words $\pi$ is a path in $\game[\sigma, v_0]$.
Since $\game[\sigma, v_0]$ is a graph satisfying $\Omega$ of size at most $n$, 
we have that $\col(\pi) = c_0 \dots c_i$ is a prefix of a word in $\Omega^{\mid n} \subseteq L(\auto)$.
This implies that $\sigma_\auto(c_0 \dots c_i)$ is well defined, hence so is $\sigma'(\pi')$,
thus $\last(\pi')$ cannot be a sink.
\end{proof}

\subsection{Universal graphs}
\label{subsec:universal}

In comparison with separating and GFSG automata, the notion of universal graphs is more combinatorial than automata-theoretic, hence the name.
However, recall that in our formalism, non-deterministic automata and graphs are syntactically the same objects.

\paragraph*{Graph homomorphisms}
For two graphs $G,G'$, a homomorphism $\phi : G \to G'$ maps the vertices of $G$ to the vertices of $G'$ such that 
\[
(v,c,v') \in E \ \implies\ (\phi(v),c,\phi(v')) \in E'.
\]

As a simple example, note that if $G$ is a subgraph of $G'$, then the identity is a homomorphism $G \to G'$.
We say that $G$ maps into $G'$ if there exists a homomorphism $G \to G'$.

\begin{fact}
If $G$ maps into $G'$, then $L(G) \subseteq L(G')$.
\end{fact}

\begin{defi}[Universal graphs]
A graph $\auto$ is $(n,\Omega)$-universal if it satisfies $\Omega$ and all graphs of size at most $n$ satisfying $\Omega$ map into $\auto$.
\end{defi}

\begin{thm}
Let $\Omega$ be a positionally determined objective, $\auto$ an $(n,\Omega)$-universal graph, $\game$ a game of size $n$ with objective $\Omega$, and $v_0 \in V$.
Then Eve has a winning strategy in $\game$ from $v_0$ if and only if she has a winning strategy in $\game \chain \auto$ from $(v_0,q_0)$ 
for some $q_0 \in Q$.
\end{thm}

\begin{proof}
Again, the ``if'' is a direct application of Lemma~\ref{lem:product_harder}.
Conversely, assume that Eve wins $\game$ from $v_0$, and let $\sigma$ be a winning strategy from $v_0$.
Since $\Omega$ is positionally determined, we can choose $\sigma$ to be positional.
Consider the graph $\game[\sigma,v_0]$: it has size at most $n$ and since $\sigma$ is winning from $v_0$ it satisfies $\Omega$.
Because $\auto$ is an $(n,\Omega)$-universal graph there exists a homomorphism $\phi$ from $\game[\sigma,v_0]$ to~$\auto$. 
We construct a positional strategy $\sigma'$ in $\game \chain \auto$ 
by playing as in $\sigma$ in the first component and following the homomorphism $\phi$ on the second component. Formally,
\[
\begin{array}{llll}
\sigma'(v,q) &=& (\sigma(v),q) & \text{ for } v \in \VE \text{ and } q \in Q, \\
\sigma'(e,q) &=& (v', \phi(v')) & \text{ for } e=(v,c,v') \in E
 \text{ and } q \in Q.
\end{array}
\]
Let $q_0 = \phi(v_0)$. We first obtain the following simulation property by a straightforward induction: 
\begin{center}
Let $\pi' = (v_0,q_0)(e_0,q_0) \dots (v_i,q_i)(e_i,q_i)$ with $e_j = (v_j,c_j,v_{j+1})$ for all $j \le i$ be consistent with $\sigma'$, then \\
$\pi = v_0 c_0 \dots v_i c_{i+1} v_{i+1}$ is consistent with $\sigma$ and for all $j \le i$ we have 
$v_i \in V[\sigma, v_0], e_i \in E[\sigma, v_0]$, and $q_i = \phi(v_i)$.
\end{center}

We now show that $\sigma'$ is a winning strategy from $(v_0,q_0)$ in $\game \chain \auto$.
Let $\pi'$ be a maximal consistent path with $\sigma'$, which we assume for contradiction to be losing, hence finite and ending in a sink controlled by Eve.
There are two cases: either $\last(\pi') = (v_i,q_i) \in \VE \times Q$, or $\last(\pi') = (e_i,q_i) \in E \times Q$.
In the first case, thanks to the simulation property this implies that $v_i$ is a sink, which contradicts that $\sigma$ is winning.
In the second case, let us note that since $\phi$ is a homomorphism and $e_i = (v_i,c_i,v_{i+1}) \in E[\sigma, v_0]$, 
we have $(\phi(v_i),c_i,\phi(v_{i+1})) \in \Delta$, 
so $\sigma'(e_i,q_i)$ is an edge from $(e_i,q_i)$ in $\game \chain \auto$,
hence $(e_i,q_i)$ is not a sink, a contradiction.
\end{proof}

\section{Equivalence via saturation}
\label{sec:saturation}
In Section~\ref{sec:reductions} we have defined three sufficient conditions on $\auto$ for the equivalence
between the games $\game \chain \auto$ and $\game$: separating automata, GFSG automata, and universal graphs.
What they have in common is that they all satisfy
\[
\Omega^{\mid n} \subseteq L(\auto) \subseteq \Omega,
\]
although this condition does not ensure the equivalence between $\game \chain \auto$ and $\game$.
In this section, we prove the following equivalence result:

\begin{thm}[Normalisation and equivalence between all three sufficient conditions]
\label{thm:main}
Let $\Omega$ be a positionally determined and prefix independent objective with a neutral letter, and let $n \in \N$. 
The smallest size of an $(n, \Omega)$-separating automaton, of a $(n,\Omega)$-GFSG automaton, of a $(n,\Omega)$-universal graph, and, more generally, of a non-deterministic automaton $\auto$ such that 
$\Omega^{|n} \subseteq L(\auto) \subseteq \Omega$, 
all coincide.

Moreover, there exists a graph $G$, whose size matches this minimal size, which is saturated and linear
(see Subsection~\ref{subsec:saturated_graphs_and_linear_graphs}), $(n,\Omega)$-universal, $(n,\Omega)$-GFSG, and such that removing edges from $G$ produces an $(n,\Omega)$-separating automaton. 
\end{thm}


We will study graphs which are saturated, that is, maximal with respect to edge inclusion for the property of satisfying $\Omega$ (see next subsection for formal definition).
We will show that these graphs have a very constrained structure captured by the notion of linear graphs, which in particular can be seen as deterministic automata.

Beyond proving the above theorem, this strong structural result will allow us to obtain generic value iteration algorithms (Section~\ref{sec:value_iteration}), and moreover yields a generic normalisation procedure (saturation), which will be applied to parity and mean payoff objectives (Sections~\ref{sec:parity} and~\ref{sec:mean_payoff}) to prove strong upper and lower bounds on the size of universal graphs.

\subsection{Saturated graphs and linear graphs}
\label{subsec:saturated_graphs_and_linear_graphs}

We introduce the two most important concepts of Section~\ref{sec:saturation}: saturated graphs and linear graphs.
Given a graph $G$ and $e \notin E$, we let $G_e$ denote the graph $(V,E \cup \set{e})$.

\begin{defi}[Saturated graphs]
A graph $G$ is saturated with respect to $\Omega$ if $G$ satisfies $\Omega$ and for all $e \notin E$ we have that 
$G_e$ does not satisfy $\Omega$.
\end{defi}

For a graph $G$ satisfying $\Omega$, we say that $\widehat{G}$ is a saturation of $G$ if 
$\widehat{G}$ is saturated, uses the same set of vertices as $G$, and all edges of $G$ are in $\widehat{G}$.

\begin{lem}
\label{lem:saturation}
Let $G$ be a graph satisfying $\Omega$. 
Then there exists a saturation $\widehat{G}$ of $G$, and the identity is a homomorphism from $G$ to $\widehat{G}$.
\end{lem}

\begin{proof}
There are finitely many graphs over the set of vertices of $G$ which include all edges of $G$, so there exists a maximal one with respect to edge inclusion for the property of satisfying $\Omega$. 
\end{proof}

The following two facts will be very useful for reasoning, they essentially say that in the context of universal graphs, 
we can restrict our attention to saturated graphs without loss of generality.

\begin{fact}
Let $n \in \N$, $\Omega$ an objective, and $\U$ a graph.
\begin{itemize}
	\item If $\U$ is an $(n,\Omega)$-universal graph, then any saturation of $\U$ is $(n,\Omega)$-universal (and it has the same size as $\U$),
	\item the graph $\U$ is $(n,\Omega)$-universal if and only if 
	it satisfies $\Omega$ and all saturated graphs $G$ of size $n$ map into $\U$. 
\end{itemize}
\end{fact}
Both are consequences of the following observation: if a saturation $\widehat{G}$ of $G$ maps into $\U$,
then $G$ maps into $\U$ by composing the homomorphisms from $G$ to $\widehat{G}$ and from $\widehat{G}$ to $\U$.

We will see that saturated graphs satisfy structural properties captured by the following definition.

\begin{defi}[Linear graphs]
A graph $G$ is linear if there exists a total order $\le$ on the vertices of $G$ satisfying the following two properties:
\begin{eqnarray*}
\text{ if } v \ge v' \text{ and } (v',c, v'') \in E, \text{ then } (v, c, v'') \in E, \\
\text{ if } (v ,c , v') \in E \text{ and } v' \ge v'', \text{ then } (v, c, v'') \in E.
\end{eqnarray*}
We refer to the first property as left composition and to the second as right composition.
\end{defi}

Given a linear graph $G=(V,E)$, a vertex $v \in V$ and a colour $c \in C$, one may define the set of $c$-successors of $v$, $\{v' \in V \mid (v,c,v') \in E\}$, which by right composition is downward closed with respect to the order on $G$. 
Hence it is uniquely defined by its maximal element, which we refer to as $\delta(v,c) \in V \cup \{\bot\}$.
Likewise, the set of $c$-predecessors of $v' \in V$, given by $\{v \in V \mid (v,c,v') \in E\}$ is upwards closed by left composition. We let $\rho(v',c) \in V \cup \{\top\}$ denote the minimal $c$-predecessor of $v'$.

For a fixed $c \in C$, the functions $v \mapsto \delta(v,c)$ and $v' \mapsto \rho(v',c)$ are non-decreasing, respectively thanks to left and right composition. If it holds for some $v,v'$ that for all colours $c$, we have $\delta(v,c)=\delta(v',c)$, and $\rho(v,c) = \rho(v',c)$, then $v, v'$, and any vertex $v''$ between $v$ and $v'$ all have exactly the same incoming and outgoing edges.

%
%
%


The operator $\delta(\cdot,c)$ also allows us to determinise linear graphs while preserving the state space.

\begin{defi}[Determinisation of linear graphs]
Let $G=(V,E)$ be a linear graph with respect to $\leq$. 
We define its determinisation $\Det(G)$ by $\Det(G) = (V, \delta, v_0)$ with $v_0$ the maximum element of $V$ 
and where $\delta(v,c)$ is defined as the largest $c$-successor of $v$ if there is any (and undefined otherwise).
\end{defi}

\begin{lem}\label{lem:determinisation}
Let $G=(V,E)$ be a linear graph. Then $G$ and $\Det(G)$ recognise the same language.
\end{lem}

\begin{proof}
It is clear that $L(\Det(G)) \subseteq L(G)$. Conversely, let $\pi = v_0 c_0 v_1 c_1 \dots$ be a path in $G$. 
We prove by induction that for all $i$, $\delta^*(v_0,c_0 \dots c_i)$ is well defined and greater than or equal to $v_i$. 
This is clear for $i = 0$, so let us consider $i > 0$. 
Since $(v_i,c_i,v_{i+1}) \in E$ and by induction hypothesis $\delta^*(v_0,c_0 \dots c_i) \ge v_i$,
we have by left composition that $(\delta^*(v_0,c_0 \dots c_i), c_i, v_{i+1}) \in E$. 
Hence $\delta^*(v_0,c_0 \dots c_{i+1})$ is well defined and greater than or equal to $v_{i+1}$, which concludes the inductive proof. 
This implies that $c_0 c_1 \dots \in L(\Det(G))$.
\end{proof}

\subsection{Normalisation theorem}
\label{subsec:saturated_graphs_are_linear_graphs}

We now prove our normalisation result.

\begin{thm}
\label{thm:saturated_graphs_are_linear}
Let $\Omega$ be a prefix independent and positionally determined objective with a neutral letter $0$. Then any saturated graph with respect to $\Omega$ is linear.
\end{thm}

This theorem is the only place in the paper where we use the neutral letter.

%
%

\begin{proof}
Let $G=(V,E)$ be a saturated graph, and consider the relation $\leq$ over $V$ given by $v \geq v' \iff (v,0,v') \in E$. We now prove that $\leq$ is a total preorder satisfying left and right composition.
\begin{itemize}
	\item \textbf{Reflexivity}. Let $v \in V$. We argue that $(v,0,v) \in E$. 
	Indeed if this were not the case then adding it to $G$ would create a graph satisfying $\Omega$ by definition of $\Omega$ and neutrality of $0$, which contradicts the fact that $G$ is saturated.
	
	\item \textbf{Totality}. Assume for contradiction that neither $(v,0,v') \in E$ nor $(v',0,v) \in E$.

	We construct a game $\game$ with objective $\Omega$. See Figure~\ref{fig:question_mark_game}.
	
	\begin{figure}
	\includegraphics[width= 0.7 \linewidth]{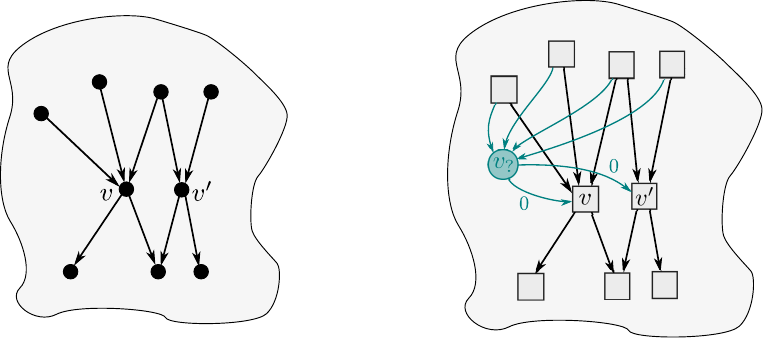}
	\caption{On the left the graph $G$, and on the right the game $\game$. The circle depicts the unique vertex which belongs to Eve, while squares belong to Adam.}
	\label{fig:question_mark_game}
	\end{figure}

	The underlying graph is obtained from $G$ with the following modification: 
	the set of vertices is $V$ plus a new vertex $v_?$.
	All vertices are controlled by Adam except for the new vertex $v_?$.
	For every edge leading to either $v$ or $v'$, we add a new edge going to $v_?$,
	where Eve chooses between going to $v$ or to $v'$ with colour $0$.
	Formally, the sets of vertices and edges are
	\[
	\begin{array}{lll}
	\VE & = & \{v_?\} \\
	\VA & = & V \\
	E' & = & E \cup \{(u,c,v_?) : (u,c,v) \in E \text{ or } (u,c,v') \in E\} \cup  \{(v_?,0,v), (v_?,0,v')\}.
	\end{array}
	\]
	Eve has a winning strategy from any vertex $v_0 \in V$: each time the play reaches $v_?$ from an edge $(u,c,v_?)$ such that 
	$(u,c,v) \in E$, Eve chooses to continue with $v$, and otherwise with $v'$. 
	Any path consistent with this strategy corresponds to a path in $G$,
	and since $G$ satisfies $\Omega$ this strategy is winning.
	We note that this strategy is not positional as it chooses $v$ or $v'$ depending on the incoming edge.
	
	We claim that Eve does not have a positional winning strategy in $\game$, 
	which contradicts the positional determinacy of $\Omega$.
	Indeed, if the positional strategy defined by $\sigma(v_?) = v$ was winning, 
	since any path in $G_{(v',0,v)}$ is consistent with $\sigma$ in $\game$, 
	this would imply that $G_{(v',0,v)}$ satisfies $\Omega$, contradicting the maximality of $G$.
	The same argument leads to a contradiction if $\sigma(v_?) = v'$ is a winning strategy.
	
	\item \textbf{Generalised transitivity}. 
	Note that left and right composition both generalise transitivity of $\leq$, that is, $(v,0,v') \in E$ (a) and $(v',0,v'') \in E$ (a) implies $(v,0,v'') \in E$ (c), by replacing, respectively in (b) or in (a), and in the conclusion (c), the colour 0 with any colour $c$.
	We prove left composition.
	Assume that $(v,0,v') \in E$ and $(v',c,v'') \in E$.
	We argue by contradiction that $(v,c,v'') \in E$. 
	Assuming otherwise, consider a path $\pi$ in $G_{(v,c,v'')}$. 
	Then $\pi'$ obtained from $\pi$ by replacing all occurrences of $v c v''$ by $v 0 v' c v''$ is a path in $G$ 
	so it must satisfy $\Omega$ and hence $\pi$ as well. 
	This contradicts the maximality of $G$.
	The proof of right composition follows the same lines.
\end{itemize}

Now to obtain a total order it suffices to refine $\leq$ arbitrarily over its equivalence classes, which preserves left and right compositions. Indeed, two vertices are equivalent for $\leq$ (before refining), that is, are interconnected with $0$-edges, if and only if they have exactly the same $c$-successors and $c$-predecessors for all colours $c$.
\end{proof}

\subsection{Ubiquity of linear graphs}
\label{subsec:linear_graphs_are_everywhere}

We now relate linear graphs to the models of Section~\ref{sec:reductions}.

\begin{lem}
\label{lem:linear_graphs_are_everywhere}
Let $\auto$ be a linear graph such that $\Omega^{\mid n} \subseteq L(\auto) \subseteq \Omega$. Then
\begin{itemize}
	\item $\Det(\auto)$ is an $(n,\Omega)$-separating automaton,
	\item $\auto$ is an $(n, \Omega)$-GFSG automaton, and
	\item $\auto$ is an $(n, \Omega)$-universal graph.
\end{itemize}
\end{lem}

\begin{proof}
By Lemma~\ref{lem:determinisation}, $L(\Det(\auto)) = L(\auto)$, so the first item is clear.
The second item follows by using the strategy induced by the determinised automaton, that is, 
$\sigma_{G}(c_0 \dots c_i) = \delta^*(q_0, c_0 \dots c_i)$.

We now focus on the third item.
Let $G=(V,E)$ be a graph satisfying $\Omega$ and of size at most $n$, and let $v \in V$.

Let $\pi = v_0 c_0 \dots c_{i} v$ be a path ending in $v$ in $G$. 
Then $\col(\pi) = c_0 \dots c_{i}$ is a prefix of a word in $\Omega^{\mid n}$, so $\delta^*(v_0, \col(\pi))$ is well defined.
We now let $\phi : V \to Q$ be defined by
\[
\phi(v) = \min \set{\delta^*(q_0, \col(\pi)) : \pi \text{ is a path in $G$ ending in $v$}}.
\]
Note that there is always a path ending in $v$ since $v$ itself is such a path.

We now prove that $\phi$ defines a graph homomorphism.
Let $(v,c,v') \in E$, and let $\pi_v$ be a path in $G$ ending in $v$ such that 
$\phi(v) = \delta^*(q_0, \col(\pi_v))$. 
Then $\pi_v c v'$ is a path in $G$ ending in $v'$ with colour $\col(\pi_v) c$, so we have 
$\phi(v') \leq \delta^*(q_0, \col(\pi_v) c) = \delta(\phi(v), c)$, and by left composition 
$(\phi(v), c, \phi(v')) \in Q$.
This shows that $\phi$ is a homomorphism from $G$ to $\auto$, so $\auto$ is indeed $(n,\Omega)$-universal.
\end{proof}

We now obtain Theorem~\ref{thm:main} as a direct consequence.

\begin{proof}
Let $\auto$ be a non-deterministic automaton such that $\Omega^{\mid n} \subseteq L(\auto) \subseteq \Omega$. Note that this includes separating automata, GFSG automata and universal graphs. Saturate $\auto$ (Lemma~\ref{lem:saturation}) to obtain (Theorem~\ref{thm:saturated_graphs_are_linear}) a linear graph. Lemma~\ref{lem:linear_graphs_are_everywhere} concludes.
\end{proof}

%

\section{Solving the product game by value iteration}
\label{sec:value_iteration}
Section~\ref{sec:saturation} shows that without increasing its size, $\auto$ can be chosen to be a linear graph. In this section, we are interested in the complexity of solving the safety game $\game \chain \auto$ under this assumption.
Fix a game $\game$ over the graph $G = (V,E)$ and a linear graph $\auto = (Q,\Delta)$.

We first give a simplification of $\game \chain \auto$ in Subsection~\ref{subsec:simplification_linear} which reduces the size of the arena to $|V| \cdot |Q| + 1$ vertices and $|E| \cdot |Q|$ edges.
Applying Theorem~\ref{thm:safety_games} yields an algorithm with runtime $O(|E| \cdot |Q|)$.
To further improve the space complexity, we first recall the (well known) linear time algorithm for solving safety games in Subsection~\ref{subsec:safety_games}, and then provide in Subsection~\ref{subsec:space-efficiency} 
a space efficient implementation of this algorithm over $\game \chain \auto$. 
This reduces the space requirement to $O(|V|)$.
Beyond space efficiency, this section provides an interesting structural insight: 
solving $\game \chain \auto$ when $\auto$ is linear precisely amounts to running a value iteration algorithm.

\begin{thm}[Complexity of solving the product game with a linear graph]
\label{thm:solving_product_game_linear}
Let $n \in \N$, $\Omega$ an objective, $\game$ a game of size $n$ with prefix independent positional objective $\Omega$, 
and $\auto = (Q, \Delta)$ a linear $(n,\Omega)$-universal graph.
Then the winning region in $\game$ can be computed~\footnote{In the RAM model with word size $w = \log(|V|) + \log(|Q|)$.} 
in time $O(|E| \cdot |Q|)$ and space $O(|V|)$.
\end{thm}

\subsection{Simplifying the product game if $\auto$ is linear}
\label{subsec:simplification_linear}

We defined the product game $\game \chain \auto$ for the general case of a non-deterministic automaton $\auto$.
If $\auto$ is linear, the product game can be simplified: from $(e,q)$ where $e=(v,c,v')$, Eve should always pick as large a successor in $Q$ as possible (if there is any successor).
More precisely, we let $\delta(q,c) = \max\{q' \in Q \mid (q,c,q') \in \Delta\}$ 
and consider the safety game $\game \mchain \auto$ with vertices $V' = \VE' \uplus \VA'$ and edges $E'$ given by
\[
\begin{array}{lll}
\VE' & = & \VE \times Q\ \cup \set{\bot}, \\
\VA' & = & \VA \times Q, \\
E' & = & \{((v,q), (v',\delta(q,c))) : q \in Q, (v,c,v') \in E, \delta(q,c) \text{ is defined}\} \\
 & \cup & \{((v,q), \bot) : q \in Q, (v,c,v') \in E, \delta(q,c) \text{ is not defined}\}.
\end{array}
\]

In words, from $(v,q) \in V \times Q$, the player whom $v$ belongs to chooses an edge $e=(v,c,v') \in E$, and the game progresses to $(v',\delta(q,c))$ if $q$ has an outgoing edge with colour $c$ in $\auto$, and to $\bot$ otherwise, which is losing for Eve.

It is clear that Eve wins from $(v,q)$ in $\game \chain \auto$ (as defined in Subsection~\ref{subsec:product}) 
if she wins from $(v,q)$ in $\game \mchain \auto$, simply by always choosing $\delta(q,c)$ as a successor from $(e,q)$ where $e = (v,c,v')$.
The converse is not much harder: just as in the proof of Lemma~\ref{lem:determinisation}, by an easy induction, picking the largest successor ensures to always remain larger on the second component, which allows for more possibilities (and in particular, avoids sinks) thanks to left composition.

\begin{rem}\label{rmk:upward-closed}
This also proves that if $(v,q)$ is in the winning region of $\game \chain \auto$ and $q' \geq q$, then $(v,q')$ is also winning. In particular, $v$ is winning in $\game$ if and only if there exists $q \in Q$ such that $(v,q)$ is winning in $\game \chain \auto$ if and only if $(v, \max Q)$ is winning in $\game \chain \auto$.
\end{rem}

From here on, since we solve $\game \chain \auto$ for linear $\auto$, we shall use the new (simpler) definition $\mchain$ instead of $\chain$.

\subsection{Solving safety games in linear time}
\label{subsec:safety_games}

We recall the standard linear time algorithm for safety games. Let $\game$ be a safety game, and consider the operator $\Safe$ over subsets of $V$ given by
\[
\begin{array}{lll}
\Safe(X) & = & \{v \in \VE \mid \exists v', (v,v') \in E\text{ and } v' \in X\} \\
& \cup & \{v \in \VA \mid \forall v', (v,v') \in E \implies v' \in X\},
\end{array}
\]
that is, $\Safe(X)$ is the set of vertices from which Eve can ensure to remain in $X$ for one step. 
Note that $\Safe(X)$ always includes sinks controlled by Adam and excludes sinks controlled by Eve,
and it is monotone: if $X \subseteq Y$, then $\Safe(X) \subseteq \Safe(Y)$.

It is well known that the winning region in $\game$ is the greatest fixpoint of $\Safe$. This can be computed via Kleene iteration by setting $X_0=V$ and $X_{i+1} = \Safe(X_i)$, which stabilizes in at most $|V|$ steps. Since computing $\Safe$ requires $O(|E|)$ operations, a naive implementation yields runtime $O(|V| \cdot |E|)$.

To obtain linear runtime $O(|E|)$, let us define $R_i = X_{i} \setminus X_{i+1}$, the set of vertices which are removed upon applying $\Safe$ for the $(i+1)$\textsuperscript{-th} time (see Figure~\ref{fig:safety_game}).
Then $R_0$ is the set of sinks controlled by Eve, and moreover we show that for $i \geq 1$, $R_i$ can be computed efficiently from $R_{i-1}$, provided one maintains a data structure $\cnt$, which stores, for each $v \in \VE$, the number of edges from $v$ to $X_i$. This is based on the observation that a vertex in $\VA \cap X_i$ belongs to $R_i$ if and only if it has an edge towards $R_{i-1}$, and a vertex in $\VE \cap X_i$ belongs to $R_i$ if and only if it has no outgoing edge towards $X_i$, and moreover it has an edge towards $R_{i-1}$ (otherwise, it would not belong to $X_i$).

Hence to compute $R_i$ and update $\cnt$, it suffices to iterate over all edges $e=(v,v')$ such that $v' \in R_{i-1}$ and $v \in X_i$: if $v \in \VA$ then $v \in R_i$, and if $v \in \VE$, then $\cnt(v)$ is decremented, and $v$ is added to $R_i$ if 0 is reached, that is, $v$ has no edge going to $X_i$. We give a pseudo-code in Appendix~\ref{sec:safety} for completeness.
The overall linear complexity of the algorithm simply follows from the fact that each edge is considered at most once.

\begin{figure}
\includegraphics[width = 0.8 \linewidth]{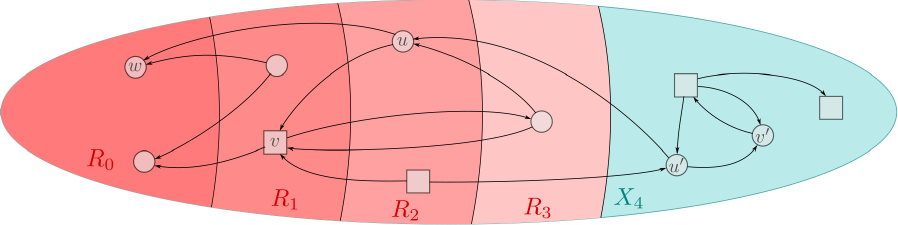}
\caption{Example of a safety game, where Eve vertices are depicted with circles and Adam vertices with squares. The iteration terminates after 4 steps: $X_4 = \Safe(X_4)$. In this example, $\cnt(u)$ is initialized at two, decremented when $w$ is added to $R_0$, and decremented again, reaching 0, when $v$ is added to $R_1$. Hence $u$ is then added to $R_2$. In contrast, $\cnt(u')$ never reaches 0 thanks to the edge towards $v'$.}
\label{fig:safety_game}
\end{figure}

\subsection{Instantiating to $\game \chain \auto$: a space efficient implementation}
\label{subsec:space-efficiency}

We now instantiate the previous algorithm for safety games to $\game \chain \auto$ 
as defined in~\ref{subsec:simplification_linear}.

For a vertex $v \in V$ and a set $X$ of vertices of $\game \chain \auto$, we let $X^v = \set{q \in Q \mid (v,q) \in X}$. 
As we have discussed (see Remark~\ref{rmk:upward-closed}), the winning region $W$ is such that $W^v$ is upward-closed for all $v$.
We will see that when instantiating the previous algorithm to $\game \chain \auto$, that is, letting $X_0 = V(\game \chain \auto)$ and $X_{i+1} = \Safe(X_i)$, it always holds that $X_i^v$ is upward-closed for all $v$, and hence it is specified by its minimal element.
Given $X$ such that $X^v$ is upward-closed for all $v$, we let $\theta_X : V \to Q \cup \{\top\}$ be defined by $\theta(v)=\min X^v$, with value $\top$ if $X^v = \emptyset$.
We now characterize the effect of $\Safe$ over sets represented in this manner.
Recall that $\rho(q',c)=\min\{q \in Q \mid (q,c,q') \in \Delta\} \in Q \cup \{\top\}$.

\begin{lem}\label{lem:charac_update}
Let $X$ be such that $X^v$ is upward-closed for all $v$, and let $X'=\Safe(X)$. Then $X'^v$ is upward-closed for all $v$, and $\theta_{X'}$ is given by
\[
\theta_{X'}(v)=
\begin{cases}
\min\{\rho(\theta_X(v'),c) \mid (v,c,v') \in E\} $ if $ v \in \VE, \\
\max\{\rho(\theta_X(v'),c) \mid (v,c,v') \in E\} $ if $ v \in \VA.
\end{cases}
\]
\end{lem}

\begin{proof}
Note that for $q,q' \in Q$ and $c \in C$ we have
\[
\delta(q,c) \geq q' \iff (q,c,q') \in \Delta \iff q \geq \rho(q',c),
\]
where the first equivalence holds by right composition while the second holds by left composition.

Let $v \in \VE$. Then
\[
\begin{array}{lll}
(v,q) \in X' =  \Safe(X) & \iff & (v,q) \text{ has an edge to } X \text{ in } \game \chain \auto \\
& \iff & \exists (v,c,v') \in E, (v', \delta(q,c)) \in X \\
& \iff & \exists (v,c,v') \in E, \delta(q,c) \geq \theta(v') \\
& \iff & \exists (v,c,v') \in E, q \geq \rho(\theta(v'),c)
\end{array}
\] 
Hence, $X'^v$ is upward-closed, and given by
\[
\theta_{X'}(v) = \min \set{q \in Q \mid (v,q) \in \Safe(X)} = \min \set{\rho(\theta(v'),c) : (v,c,v') \in E}.
\]
For $v \in \VA$ replacing $\exists$'s by $\forall$'s in the previous chain of equivalences yields 
\[
\theta_{X'}(v) = \min \set{q \in Q \mid (v,q) \in \Safe(X)} = \max \set{\rho(\theta(v'),c) : (v,c,v') \in E}.
\qedhere\]
\end{proof}

Let $X_0=V(\game \chain \auto)$ and $X_{i+1}=\Safe(X)$.
By a straightforward induction, for all $i$ and all $v$, $X_i^v$ is upward-closed, and we let $\theta_i = \theta_{X_i}$.
By the above lemma, for each $v$, $\theta_{i+1}(v)$ can be computed from the value of $\theta_i$ at each successor $v'$ of $v$.
Note that by monotonicity of $\Safe$, for each $v$ we have $\theta_{i+1}(v) \geq \theta_i(v)$.
Following the usual terminology for value iteration algorithms, we say that a vertex $v$ is invalid at step $i$ if $\theta_{i+1}(v) > \theta_{i}(v)$, or equivalently, $(v,\theta_i(v)) \in X_i \setminus \Safe(X_i)$.
We let $\Inv_i \subseteq V$ denote the set of vertices which are invalid at step $i$.
To compute $\theta_{i+1}$ from $\theta_i$, it suffices to update the value of $\theta_i(v)$ for each $v \in \Inv_i$. See Figure~\ref{fig:product_safety_game}.

\begin{figure}
\includegraphics[width = 0.8 \linewidth]{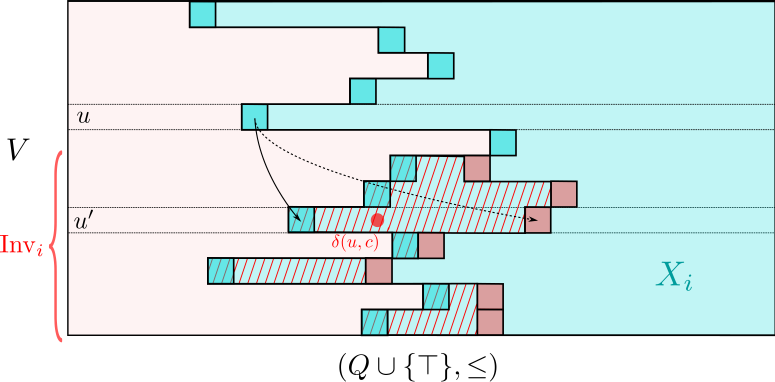}
\caption{Computing the $i$-th layer in the safety game $\game \chain \auto$. For each $v$, $X_i^v$ is upward-closed, and $X$ is represented by $\theta:V \to Q \cup \{\top\}$. The blue squares represent vertices of the form $(v,\theta(v))$, which lie at the border of $X_i$. The hatched area represents the $i$-th layer $X_i \setminus X_{i+1}$. Computing the position of the red squares, given by $\theta'$, is achieved using Lemma~\ref{lem:charac_update}. Computing the new set $\Inv_{i+1}$ of invalid vertices is achieved by inspecting $c$-predecessors $u$ in $\game$ of each $u' \in \Inv_i$: if $\theta(u') \leq \delta(u,c) < \theta'(u')$, then there was an edge in $\game \chain \auto$ from $(u,\theta(u))$ towards $X_i^{u'}$, but there is no edge from $(u,\theta'(u))$ to $X_{i+1}^{u'}$. Hence, if $u \in \VA$ it is added to $\Inv_{i+1}$, and if $u \in \VE$, $\cnt(u)$ is decremented, and $u$ is added to $\Inv_{i+1}$ if 0 is reached.}
\label{fig:product_safety_game}
\end{figure}

Following the previous generic linear time algorithm for solving safety games, we store for each Eve vertex the number of edges from $(v,\theta_i(v))$ to $X_i$, which allows to efficiently compute the set of vertices which are invalid at step $i$ from those which are invalid at step $i-1$. We again give a full pseudo-code in Appendix~\ref{sec:product_safety}.

Let us conclude with analysing the complexity of the generic value iteration algorithm
in the RAM model with word size $w = \log(|V|) + \log(|Q|)$, in which computing $\delta$ and $\rho$, checking whether $q \le q'$ and evaluating $\theta$ require constant time, while storing $\theta$ requires space linear in $|V|$.

Note that at step $i$, each vertex $v \in \Inv_i$ sees a strict increase in $\theta_i(v) \in Q \cup \{\top\}$. Hence, there is at most $|Q|+1$ values of $i$ such that $v \in \Inv_i$. Moreover, computing the $i$\textsuperscript{th} update requires, for each $v \in \Inv_i$, $O(\outdeg(v))$ constant-time operations for computing $\theta'(v)$, and $O(\indeg(v) + \outdeg(v))$ constant-time operations for updating $\cnt$ (at $v$ and its predecessors) and computing $\Inv_{i+1}$.

Summing over all vertices, we obtain the wanted $O(|E| \cdot |Q|)$ time complexity.

\begin{rem}
The value iteration algorithm we present here instantiates a linear time implementation of the (global) fixpoint computation $X_{i+1} = \Safe(X)$ to the safety game $\game \chain \auto$. There exists a variation of this algorithm which also achieves linear time for solving generic safety games, and works locally: keep a queue $R$ of vertices which should be removed from $X$, and a data structure $\cnt$ storing the number of outgoing edges from $v \in \VE$ to $X$ and iteratively pick a vertex from $R$, remove it from $X$, and update $\cnt$ (potentially adding them to $R$) on each predecessor of $v$.

It is not hard to instead adapt the local variant of the algorithm to the safety game $\game \chain \auto$, which yields another slightly different generic value iteration procedure (with same complexity). The local variant is more in line with known value iteration algorithms for parity games~\cite{J00,JL17,FJSSW17} as well as mean-payoff games~\cite{BCDGR11,DKZ19}. We chose to present the global variant for its structural simplicity, and direct ties with the fixpoint iteration. For details about the local variant of the generic algorithm, we refer to~\cite{FGO18}.\end{rem}

This concludes the first part of the paper developing the theory of universal graphs.
In the second part, we will consider four (classes of) objectives and construct algorithms from universal graphs.

\section{Parity games}
\label{sec:parity}
The first objective we consider is parity. 
We refer to Subsection~\ref{subsec:complexity_parity} 
for a discussion on existing results.

The set of colours is $C = [0,d]$, where $d$ is even and we study the objective given by
\[
\Parity(d) = \{w \in [0,d]^\omega \mid \limsup w \text{ is even} \}.
\]
In the context of parity games, letters (or colours) are usually called priorities.

We give a convenient characterization of graphs which satisfy the even parity objective.
We say that a cycle is even if its maximal priority is even, and odd otherwise.
\begin{lem}{(Graphs satisfying $\Parity(d)$)}
Let $G=(V,E)$ be a $[0,d]$-graph. Then $G$ satisfies $\Parity(d)$ if and only if all cycles in $G$ are even.
\end{lem}

\begin{proof}
Assume that $G$ satisfies $\Parity(d)$.
Let $C$ be a cycle in $G$.
Repeating $C$ induces an infinite path $\pi$ in $G$ with $\col(\pi) = \col(C)^\omega$.
Then $\limsup \col(\pi) = \max \col(C)$ is even, and $C$ is even.

Conversely, assume $G$ has only even cycles, and pick an infinite path $\pi$ in $G$.
Then there is a vertex $v$ visited infinitely often by $\pi$, and then $\pi$ may be decomposed into an infinite sequence of cycles $C_1, C_2, \dots$ which all start and end in $v$.
Then $\limsup \col(\pi) = \limsup_i \col(\max(C_i))$, which is even since for all $i$, $\max(\col(C_i))$ is even.
\end{proof}

The following well known theorem states that the parity objective satisfies our assumptions.

 \begin{thmC}[\cite{Emerson&Jutla:1991,McNaughton:1993,Mostowski:1991}]
 Parity objectives are positionally determined, prefix independent, and have 0 as a neutral letter.
 \end{thmC}

\subsection{A universal graph of exponential size}

Let us fix a non-negative integer $n$.
We give a first construction of a $(n,\Parity(d))$-universal graph $\U=(U,F)$.
Its set of vertices is given by
\[
U = [0,n-1]^{d/2},
\]
that is, tuples of $d/2$ non-negative integers smaller than $n$.
Such tuples represent occurrences of odd priorities and for convenience we index with them using odd integers, and in a decreasing fashion, that is, $v \in U$ is denoted $v=(v_{d-1},v_{d-3}, \dots, v_3, v_1)$.

Following~\cite{J00}, we define an increasing sequence of total preorders 
\[
\leq_1 \ \subseteq \ \leq_3 \ \subseteq \ \dots \ \subseteq \ \leq_{d-1} \ \subseteq \ \leq_{d+1}
\]
over $U$, which are obtained by restricting to the first few values, and comparing lexicographically:
\[
v \leq_p v' \iff (v_{d-1}, v_{d-3}, \dots, v_p) \leq_{\lex} (v'_{d-1}, v'_{d-3}, \dots, v'_p).
\]
We use $v =_p v'$ to denote the equivalence relation induced by $\leq_p$, that is $v \leq_p v'$ and $v' \leq_p v$.
We write $v <_p v'$ if $\neg(v' \leq_p v)$.
Note that $\leq_{d+1}$ is the full relation: for all $v,v'$ we have $v =_{d+1} v'$.
At the other end of the scope, $\leq_1$ is antisymmetric, or in other words its equivalence classes are trivial and $=_1$ coincides with the equality.

The set of edges of $\U$ is given by
\[
(v,p,v') \in F \iff \begin{cases}
v >_p v' $ if $p$ is odd,$\\
v \geq_{p+1} v' $ if $p$ is even.$
\end{cases}
\]

Note that this condition is vacuous for $p=d$, every pair of vertices is connected (in both directions) by a $d$-edge.

\begin{lem}
\label{lem:first_construction_for_parity}
The graph $\U$ is $(n, \Parity(d))$-universal.
\end{lem}

\begin{proof}
We first show that $\U$ satisfies $\Parity(d)$.
Let $\pi=v_0 p_0 v_1 p_1 \dots$ be an infinite path in $\U$, and assume for contradiction that $p = \limsup_i p_i$ is odd: for all large enough $i$ we have $p_i \leq p$, and for infinitely many $i$, $p_i = p$.
Then for all large enough $i$ we have $(v_{i,d-1}, v_{i,d-3}, \dots v_{i,p}) \geq_\lex (v_{i+1,d-1}, v_{i+1,d-3}, \dots, v_{i+1,p})$ and this inequality is strict for infinitely many $i$'s, a contradiction.

We now show that $\U$ embeds all graphs of size $\leq n$ which satisfy $\Parity(d)$.
Let $G=(V,E)$ be such a graph.
Let $v_0 \in G$, let $\pi=v_0 p_0 v_1 p_1 \dots$ be a path from $v_0$ in $G$, and let $p$ be an odd priority.
Consider the number $\occ_{\pi,p}$ of occurrences of $p$ in $\pi$ before a priority greater than $p$ is seen.
We claim that $\occ_{\pi,p} < n$, by a pumping argument.
Assume for contradiction that $\occ_{\pi,p} \geq n$, and let $i_1, \dots i_n$ denote the first $n$ occurrences of $p$ in $\pi$.
We have $p_{i_1} = \cdots = p_{i_n} = p$, and for all $i \leq i_n$, $p_i \leq p$.
Then the $n+1$ vertices $v_{0},v_{i_1+1},v_{i_2+1}, \dots, v_{i_n+1}$ in this order, all have a path with maximal priority $p$ to the next one.
Since $G$ has $n$ vertices, there must be a repetition in this sequence, which induces an odd cycle.
Given a path $\pi$, we let $\occ_\pi = (\occ_{\pi,d-1}, \occ_{\pi, d-3}, \dots, \occ_{\pi,1})$.

We now define $\phi : V \to U$ by $\phi(v) = \max_\pi \occ_\pi$, where the max is taken lexicographically over all paths starting in $v$.
We claim that $\phi$ defines a graph homomorphism.
Let $e=(v,p,v')$ be an edge in $G$ and let $\pi'$ be a path from $v'$ in $G$ with maximal $\occ$.
Then $\pi = v p \pi'$ is a path from $v$ in $G$, satisfying $\occ_{\pi, p} = \occ_{\pi', p} + 1$ if $p$ is odd, and for all odd $\ell > p$, $\occ_{\pi, \ell} = \occ_{\pi', \ell}$.
If $p$ is odd this implies $\occ_{\pi} >_p \occ_{\pi'} = \phi(v')$, and if $p$ is even we have $\occ_{\pi} \geq_{p+1} \occ_{\pi'} = \phi(v')$.
In both cases, this yields $(\phi(v),p,\phi(v')) \in F$ since $\phi(v) \geq_\lex \occ_{\pi}$.
\end{proof}

The set $[0,n-1]^{d/2}$ of vertices of $\U$ can be seen as the set of leaves of the (ordered) complete tree of height $d/2+1$ and of degree $n$.
In this point of view, the equivalence relations $=_{2k+1}$ groups together leaves which belong to the same subtree at level $k$ (where by convention, the leaves are at level 0), hence the equivalence classes correspond to node of level $k$, which are naturally ordered left-to-right by $\leq_{2k+1}$.
In other words, a sequence of preorders $\leq_1 \subseteq \leq_3 \subseteq \dots \subseteq \leq_{d+1}$ over some set $V$ naturally induces a tree structure of height $d/2+1$ with leaves $V$.

It is not hard to see that $\U$ is saturated for the parity objective: the addition of any edge induces an odd cycle.
We will actually prove that every saturated graph is given by such a sequence of preorders, which induces the structure of a tree.
This will allow us to reformulate the universality condition over parity graphs as a universality condition over trees.

\subsection{Trees, tree-like graphs, and saturated parity graphs}
\label{subsec:saturated_graphs_parity}

Towards describing the structure of saturated graphs for the parity objective, we now define (ordered, levelled) trees.

A tree of height $h \geq 0$ is a finite subset $T$ of $\Z^h$.

As previously, we think of elements of a tree as representing occurrences of odd priorities. In this regard it is convenient to use odd numbers as indices, and we shall write $v=(v_{d-1}, v_{d-3}, \dots, v_{3}, v_1) \in T$, where $d=2h$.

A tree of height $h$ naturally defines an increasing sequence of $h+1$ total preorders
\[
\leq_1 \subseteq \leq_3 \subseteq \dots \subseteq \leq_{d-1} \subseteq \leq_{d+1},
\]
given by
\[
(v_{d-1}, v_{d-3}, \dots, v_{1}) \leq_p (v'_{d-1}, v'_{d-3}, \dots, v'_{1}) \iff (v_{d-1}, \dots, v_p) \leq_{\lex} (v'_{d-1}, v'_{d-3}, \dots, v'_p).
\]
Again, we use the standard notations relative to total preorders, $=_p$ denotes the equivalence relation associated with $\leq_p$, whose equivalence classes are stricly ordered by $<_p$.

\begin{figure}[ht]
\centering
\includegraphics[width=.22\linewidth]{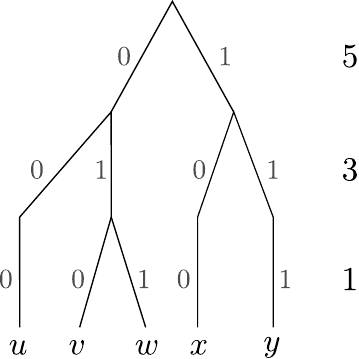}
\caption{The tree $T=\{000, 010, 011, 100, 111\}=\{u,v,w,x,y\}$. It has height $h=3$. We have $u =_5 v =_5 w <_5 x =_5 y$, and $u <_3 v =_3 w <_3 x <_3 y$.}
\label{fig:example_tree}
\end{figure}

Note that $\leq_{d+1}$ is the full relation, it has one equivalence class.
The order $\leq_1$ is antisymmetric; it is a total order, which coincides with $\leq_\lex$, and $=_1$ coincides with the equality over $T$.
Note that there is a unique tree of height $0$ which we call the empty tree, and that trees of height $1$ are sets of integers.

A tree $T$ of height $h$ induces a $[0,d]$-graph $\Graph(T)$ with vertices $T$ defined as follows:
\begin{itemize}
\item for even $p$, $(v,p,v') \in E$ if and only if $v \geq_{p+1}^T v'$, and
\item for odd $p$, $(v,p,v') \in E$ if and only if $v >_{p}^T v'$.
\end{itemize}

We say that graphs of the form $\Graph(T)$, where $T$ is a tree, are tree-like.

\begin{lem}
\label{lem:trees_satisfy_parity}
Let $d \in \N$ be an even number and $T$ a tree of height $d/2$, then the graph $\Graph(T)$ satisfies $\Parity(d)$, and moreover, it is saturated.
\end{lem}

\begin{proof}
We first show that $\Graph(T)$ satisfies $\Parity(d)$, which is a direct adaptation of the first part of the proof of Lemma~\ref{lem:first_construction_for_parity}. We spell it out for completeness.
Let $\pi=v_0 p_0 v_1 p_1 \dots$ be an infinite path in $\Graph(T)$, and assume for contradiction that $p = \limsup_i p_i$ is odd: for all large enough $i$ we have $p_i \leq p$, and for infinitely many $i$, $p_i = p$.
Then for all large enough $i$ we have $(v_{i,d-1}, v_{i,d-3}, \dots v_{i,p}) \geq_\lex (v_{i+1,d-1}, v_{i+1,d-3}, \dots, v_{i+1,p})$ and this inequality is strict for infinitely many $i$'s, a contradiction.

We now show that adding any edge to $\Graph(T)$ yields a graph not satisfying $\Parity(d)$.
Let $e = (v,p,v')$ be an edge that does not appear in $\Graph(T)$. 
If $p$ is even, then $v' >_{p+1} v$, 
so the edge $(v',p+1,v)$ belongs to $\Graph(T)$, hence $\Graph(T)_e$ contains an odd cycle. 
If $p$ is odd, then $v' \geq_{p} v$, which implies that the edge $(v',p-1,v)$ belongs to $\Graph(T)$, 
and again $\Graph(T)_e$ contains an odd cycle. 
\end{proof}

Now the second part of the proof of Lemma~\ref{lem:first_construction_for_parity} can be re-interpreted as follows.

\begin{lem}
\label{lem:universality_of_treelike}
Let $G=(V,E)$ be a graph satisfying $\Parity(d)$, where $d$ is even. There exists a tree $T$ of height $d/2$ and size $\leq |V|$ such that $G$ embeds in $\Graph(T)$.
\end{lem}

\begin{proof}
We precisely mimic the proof of Lemma~\ref{lem:first_construction_for_parity}.
Given a path $\pi$ in $G$ and an odd priority $p \in \{1,3, \dots, d-1\}$, we again define $\occ_{\pi,p}$ to be the number of occurrences of $p$ in $\pi$ before a priority greater than $p$ is seen, and we let $\occ_{\pi}=(\occ_{\pi,d_1}, \occ_{\pi,d-3}, \dots, \occ_{\pi,1})$.
It is clear, since $G$ satisfies $\Parity(d)$, that all $\occ_{\pi,p}$'s are finite.
For $v \in V$, we define $\phi(v)$ to be the lexicographic maximum of $\occ_\pi$'s over paths $\pi$ starting in $v$, and put $T = \phi(V)$.

%

We prove that $\phi$ defines a homomorphism from $G$ to $\Graph(T)$.
Let $e=(v,p,v')$ be an edge in $G$ and let $\pi'$ be a path from $v'$ in $G$ with maximal $\occ$.
Then $\pi = v p \pi'$ is a path from $v$ in $G$, satisfying $\occ_{\pi, p} = \occ_{\pi', p} + 1$ if $p$ is odd, and for all odd $\ell > p$, $\occ_{\pi, \ell} = \occ_{\pi', \ell}$.
If $p$ is odd this implies $\occ_{\pi} >_p \occ_{\pi'} = \phi(v')$, and if $p$ is even we have $\occ_{\pi} \geq_{p+1} \occ_{\pi'} = \phi(v')$.
In both cases, this yields $(\phi(v),p,\phi(v')) \in F$ since $\phi(v) \geq_\lex \occ_{\pi}$.
\end{proof}

Actually, it can be shown that, modulo contracting equivalence classes of $0$-edges which have the same ingoing and outgoing edges, saturated graphs are exactly tree-like graphs.
This is non-essential for what follows, but we prove it as an interesting side remark.

\begin{lem}
\label{lem:saturated_graphs_treelike}
Let $G=(V,E)$ be a saturated graph with respect to $\Parity(d)$, where $d$ is even, and such that for all $v,v'$, if both $(v,0,v')$ and $(v',0,v)$ are edges in $G$ then $v=v'$.
Then there exists a tree $T$ of height $d/2$ such that $G = \Graph(T)$ up to renaming the vertices.
\end{lem}

\begin{proof}
By Lemma~\ref{lem:universality_of_treelike}, there exist $T$ of height $d/2$ and size $\leq |V|$ and a homomorphism $\phi$ from $G$ to $\Graph(T)$.
It suffices to prove that $\phi$ is injective.
Indeed, this ensures that the sets of vertices of $G$ and $\Graph(T)$ are in bijection; moreover an edge $(\phi(v),c,\phi(v'))$ belongs to $\Graph(T)$ whenever $(v,c,v')$ belong to $G$ since $\phi$ is a homomorphism, and the converse is true by saturation.

Let $v,v' \in V$ be such that $\phi(v)=\phi(v')$.
We show that $e=(v,0,v')$ and $e'=(v',0,v)$ belong to $E$, which implies that $v=v'$.	
Assume $e$ and $e'$ do not both belong to $E$.
Since $G$ is saturated, there exists an infinite path $\pi=v_0 p_0 v_1 p_1 \dots$ in $G_{e,e'}$ with $\limsup_i p_i$ odd.
But by definition of $\Graph(T)$, its vertices, and in particular $\phi(v)$, have $0$-self-loops, which implies that $\phi(e)=\phi(e')$ are edges in $\Graph(T)$.
Hence, $\phi(\pi)= \phi(v_0) p_0 \phi(v_1) p_1 \dots$ is a path in $\Graph(T)$ with odd limsup, a contradiction.
\end{proof}

Interpreting saturated graphs as trees allows us to rephrase the universal property for graphs as a simpler one for trees.

We say that a map $\phi: T \to T'$ where $T$ and $T'$ are trees of height $h$ is a tree-homomorphism if it preserves all orderings, that is, for all $v,v' \in T$, and for all odd $p$, 
\[
v \leq^T_p v' \iff \phi(v) \leq^{T'}_p \phi(v').
\]
An example is depicted in Figure~\ref{fig:example_embedding}.

\begin{lem}
\label{lem:equivalence_homomorphisms}
Let $T$ and $T'$ be two trees of height $d/2$, and let $\phi:T \to T'$.
Then $\phi$ is a tree-homomorphism if and only if it is a graph homomorphism from $\Graph(T)$ to $\Graph(T')$.
\end{lem}

\begin{proof}
Recall that $v <_p v'$ is the negation of $v \geq_p v'$. Hence we have that
\[
\begin{array}{ccl}
& & \phi $ is a tree-homomorphism $ \\[\medskipamount]
& \iff & \forall v,v' \in V, \forall p$ odd, $ (v \leq^T_p v' \iff \phi(v) \leq_p^{T'} \phi(v')) \\[\medskipamount]

& \iff & \forall v,v' \in V
\left\{ \!
\begin{array}{cl}
\forall p $ even, $& (v \leq_{p+1}^{T} v' \implies \phi(v) \leq_{p+1}^T \phi(v')) \\
\forall p $ odd, $& (v <_p^{T} v' \implies \phi(v) <_{p}^{T} \phi(v'))
\end{array}
\right.\\[\bigskipamount]

& \iff & \forall v,v' \in V, \forall p, ((v',p,v) \in E(\Graph(T)) \implies (\phi(v'),p,\phi(v) \in E(\Graph(T')))\\[\medskipamount]
& \iff & \phi$ is a graph homomorphism,$
\end{array}
\]
the wanted result.
\end{proof}

Thanks to Lemma~\ref{lem:universality_of_treelike} and~\ref{lem:equivalence_homomorphisms}, we may now shift our attention only to trees and their homomorphisms.
We say that a tree is $(n,h)$-universal if it has height $h$ and embeds all trees of height $h$ and size $\leq n$.

\begin{figure}[!ht]
\begin{center}
\includegraphics[width=.7\linewidth]{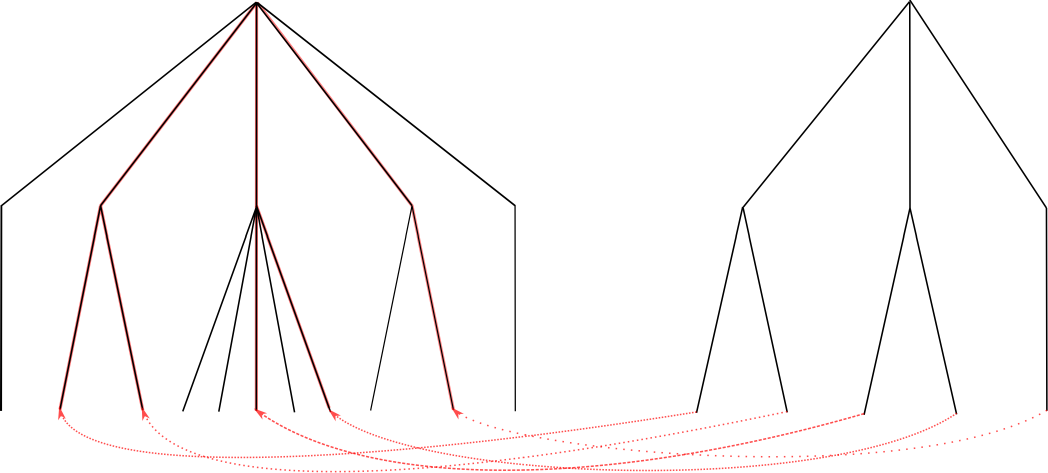}
\end{center}
\caption{On the left, a tree of height $2$ and size $11$.
This tree is $(5,2)$-universal, and is actually of minimal size.
On the right, a tree of size $5$ and one possible embedding in the universal tree.}
  \label{fig:example_embedding}
\end{figure}

\begin{lem}
The smallest $(n,h)$-universal trees and the smallest $(n, \Parity {2h})$-universal graphs have the same size.
\end{lem}

\begin{proof}
Let $G$ be a $(n, \Parity{2h}$-universal graph. By Lemma~\ref{lem:universality_of_treelike}, there is a tree $T$ of height $h$ not larger than $G$ and such that $G$ embeds into $\Graph(T)$.
Let $T'$ be a tree of height $h$ and size $n$.
By universality of $G$, there is a homomorphism from $\Graph(T')$ to $G$, which by composition induces a homomorphism from $\Graph(T')$ to $\Graph(T)$.
We conclude by Lemma~\ref{lem:equivalence_homomorphisms} that $T$ embeds $T'$, and hence $T$ is $(n,h)$-universal.

Conversely, let $T$ be $(n,h)$-universal, and let $G'$ be a graph satisfying $\Parity{2h}$.
By Lemma~\ref{lem:universality_of_treelike} there is a tree $T'$ of height $h$ such that $G'$ maps into $\Graph(T')$.
Now $T'$ maps into $T$ by universality, by Lemma~\ref{lem:equivalence_homomorphisms}, this translates into a map from $\Graph(T')$ to $\Graph(T)$, and finally by composition on the left, $G'$ maps into $\Graph(T)$.
We conclude that $\Graph(T)$ is $(n,\Parity{2h})$-universal, and it has the same size as $T$.
\end{proof}

\subsection{Upper and lower bounds on universal trees}

We now study the size of the smallest $(n,h)$-universal trees.
In this subsection, we simply index elements of trees of height $h$ with integers from $h$ to $1$ for clarity.

\begin{thm}
\label{thm:universal_trees}
Let $n, h \in \N$.
\begin{itemize}
	\item There exists a $(n,h)$-universal tree of size at most 
	\[
	2n \cdot \binom{\lceil \log(n) \rceil + h - 1}{\lceil \log(n) \rceil},
	\]
	\item All $(n,h)$-universal trees have size at least 
	\[
	\binom{\lfloor \log(n) \rfloor + h - 1}{\lfloor \log(n) \rfloor}.
	\]
\end{itemize}
\end{thm}

Let us start with the upper bound.

\begin{prop}
\label{prop:upper_bounds_universal_tree}
There exists a $(n,h)$-universal tree with size $f(n,h)$, where $f$ satisfies the following recursion
\[
\begin{array}{lll}
f(n,h) & = & f(n,h-1) + f(\lfloor n/2 \rfloor,h) + f(n - 1 - \lfloor n/2 \rfloor,h), \\
f(n,1) & = & n, \\
f(1,h) & = & 1.
\end{array}
\]
\end{prop}

We refer to Appendix~\ref{sec:analysis_F} for an analysis of this recurrence, leading to the upper bound stated
in Theorem~\ref{thm:universal_trees}.
We further analyse this function when presenting the induced algorithm in Subsection~\ref{subsec:complexity_parity}.

\begin{proof}
We construct a $(n,h)$-universal tree $T$ by induction over $(n,h)$, lexicographically.
Trees of height $1$ and size $\leq n$ are (downward-closed) subsets of $[0, n-1]$, hence $[0,n-1]$ is $(n,1)$-universal.
There is a unique tree of size $1$ and height $h$, namely, $\{0^h\}$, and it is $(1,h)$-universal.
Now, let $n$ and $h$ be $>1$ and assume constructed $(n',h')$-universal trees of size $f(n',h')$ for each $(n',h') <_\lex (n,h)$.
Let
\begin{itemize}
	\item $T_\tleft$ be a $(\lfloor n/2 \rfloor,h)$-universal tree,
	\item $T_\tmiddle$ be a $(n,h-1)$-universal tree, and
	\item $T_\tright$ be a $(\lceil n/2 \rceil - 1,h)$-universal tree.
\end{itemize}
Intuitively, we construct $T$ by merging the roots of $T_\tleft$ and of $T_\tright$ and inserting in between a child of the root to which $T_\tmiddle$ is attached.

\begin{figure}[!h]
\centering
\includegraphics[scale=.56]{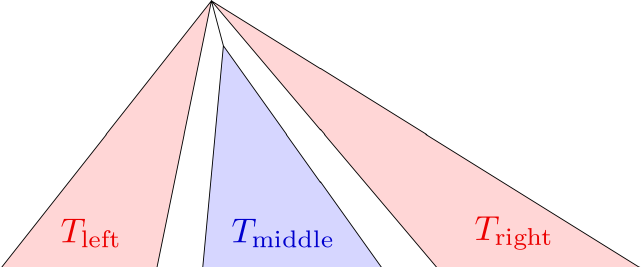}
\caption{The inductive construction.}
\label{fig:smallest_tree_construction}
\end{figure}

Formally, we let $\ell = \max_{v \in T_\tleft} v_{h}$ be the number of children of $T_\tleft$, and put
\[
T = T_\tleft \cup \ell \cdot T_\tmiddle \cup (T_\tright + (\ell +1, 0, \dots, 0)).
\] 

We now argue that $T$ is $(n,h)$-universal.
Consider a tree $T'$ of height $h$ and size $n$.
We let $p = \max_{v \in T'} v_{h}$ be the number of children of $T'$.
For each $i \in \{0, \dots, p\}$, let $n_i$ be the number of leaves in the $i$-th child of $T'$, that is, $n_i=\{v \in T' \mid v_{h} = i\}$.
Since $n_0 + n_1 + \cdots + n_p = n$, there is $i_0$ such that $n_{0} + \cdots + n_{i_0 -1} \leq \lfloor n/2 \rfloor$ and $n_{0} + \cdots + n_{i_0-1} + n_{i_0} > \lfloor n/2 \rfloor$, which implies $n_{i_0 +1} + \cdots + n_p \leq \lceil n/2 \rceil - 1$.

Consider the three trees $T'_\tleft = \{v \in T' \mid v_{h} < i_0\}$, $T'_\tmiddle = \{(v_{h-1}, \dots, v_1) \in \N^{h-1} \mid (i_0,v_{h-1}, \dots, v_1) \in T'\}$ and $T'_\tright = \{v \in T' \mid v_{d-1} > i_0\} - (i_0+1, 0, \dots, 0)$.

The trees $T'_\tleft$ and $T'_\tright$ have height $h$ and respective sizes $\leq \lfloor n/2 \rfloor$ and $\leq \lceil n/2 \rceil - 1$, hence they map respectively, via $\phi_\tleft$ and $\phi_\tright$ into $T_\tleft$ and $T_\tright$. 
Likewise, the tree $T'_\tmiddle$ has height $h-1$ and size at most $n$, hence it has a homomorphism $\phi_\tmiddle$ into $T_\tmiddle$.
These are combined into a map $\phi$ from $T'$ to $T$ by
\[
\phi(v) = \left\{ \! \begin{array}{ll}
\phi_\tleft(v) & $ if $ v_h < i_0 \\
i_0 \cdot \phi_\tmiddle(v_{h-1}, \dots, v_1) & $ if $v_h = i_0 \\
(\ell +1, 0, \dots, 0) + \phi_{\tright}(v - (i_0+1,0 \dots, 0)) & $ if $v_h > i_0. 
\end{array}
\right.
\]
We now prove that $\phi$ is indeed a tree-homomorphism from $T'$ to $T$. 
Let $k \in \{1, \dots, h\}$ and let $v,v' \in T'$.
If $v_h$ and $v'_h$ are either both $< i_0$, both $= i_0$, or both $> i_0$, then they both correspond to vertices of $T'_\tleft$ or of $T'_\tmiddle$ or of $T'_\tright$, and we conclude that $v \leq^{T'}_k v' \iff \phi(v) \leq^T_k \phi(v')$ by invoking the fact that either $\phi_\tleft, \phi_\tmiddle$ or $\phi_\tright$ is a tree-homomorphism.
Otherwise, without loss of generality $v \leq v'$, and we even have $v_h \leq i_0 \leq v'_h$ with one of these inequalities being strict, hence $v <^{T'}_k v'$, whatever the value of $k$.
By definition of $\phi$ we have $\phi(v)_h \leq \ell \leq \phi(v')_h$ with the same inequality being strict, which likewise implies that $\phi(v) <^T_k \phi(v')$, the wanted result.
Hence $T$ is universal, which concludes the proof.
\end{proof}

We now prove a lower bound on the size of universal trees.

\begin{prop}
Any $(n,h)$-universal tree has at size at least $g(n,h)$, where $g$ satisfies the following recursion
\[
\begin{array}{lll}
g(n,h) & = & \sum_{\delta = 1}^n g(\lfloor n / \delta \rfloor,h-1), \\
g(n,1) & = & n, \\
g(1,h) & = & 1.
\end{array}
\]
\end{prop}

We refer to Appendix~\ref{sec:analysis_G} for an analysis of this recurrence, leading to the lower bound stated
in Theorem~\ref{thm:universal_trees}.

The upper and lower bounds do not match perfectly.
However, 
\[
\frac{f(n,h)}{g(n,h)} \le 2^{\lceil \log(n) \rceil} \cdot \frac{\lfloor \log(n) \rfloor + h}{\lfloor \log(n) \rfloor} = O(n h),
\]
\textit{i.e.} they are polynomially related.

\begin{proof}
The bounds are clear for $h = 1$ or $n = 1$.
We let $h>1$ assumed the result known for $h-1$, and let $T$ be a $(n,h)$-universal tree, and we let $\delta \in [1,n]$. 
We construct a tree $T_\delta$ of height $h-1$ by restricting to subtrees of height 1 which have $\geq \delta$ children, that is
\[
T_\delta=\{v \in \N^{h-1} \mid |\{k \mid v \cdot k \in T\}| \geq \delta\},
\]
and argue that $T_\delta$ is $(\lfloor n / \delta \rfloor,h-1)$-universal.

\begin{figure}[ht]
\begin{center}
\includegraphics[width = 0.59 \linewidth]{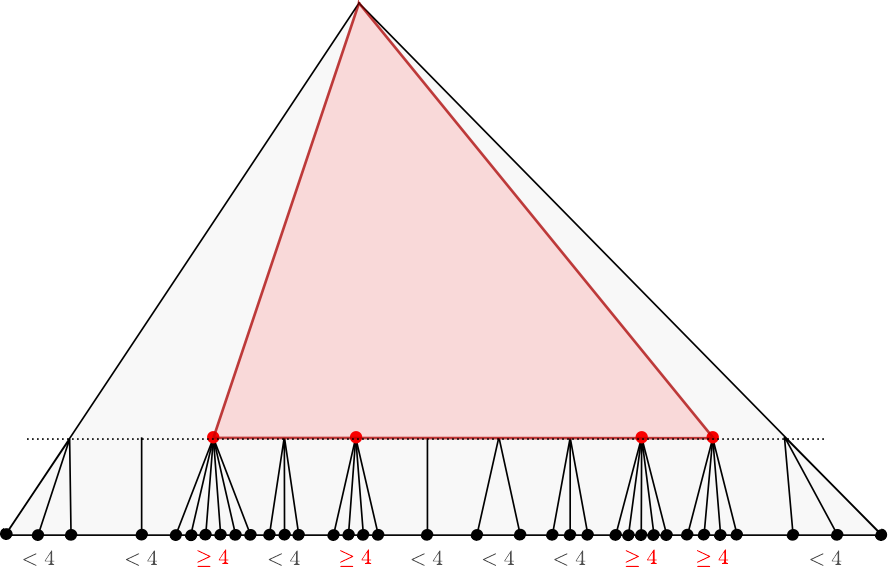}
\end{center}
\caption{In black the tree $T$ and its lowest level, and in red the tree $T_4$ obtained for $\delta=4$.}
\label{fig:parity_lower_bound_construction}
\end{figure}

Indeed, let $T'$ be a tree of height $h-1$ with $\lfloor n / \delta \rfloor$ leaves.
To each leaf of $T'$ we append $\delta$ children, that is, we consider
\[
T'' = T \times [0, \delta-1],
\]
which has size $\leq n$ and height $h$.
Since $T$ is $(n,h)$-universal, there is a tree-homomorphism $\phi$ from $T''$ to $T$.
Now, for each $v \in T'$, the elements $v \cdot 0, \dots, v \cdot (\delta -1)$ of $T''$ are different and $2$-equivalent in $T''$, hence so must be their image in $T$.
As a consequence, $\phi'(v)$ defined by $\phi'(v)=(\phi(v\cdot 0)_h, \dots, \phi(v\cdot (\delta -1)_{1})$ satisfies $|\{k \mid \phi'(v) \cdot k \in T\}| \geq k$, or in other words, $\phi'(v) \in T_\delta$, and $\phi'$ preserves all orders just because $\phi$ does.
Hence, $\phi'$ embeds $T$ in $T_\delta$, and we conclude that $T_\delta$ is $(\lfloor n / \delta \rfloor,h-1)$-universal.

In particular, the induction hypothesis tells us that $T_\delta$ has size at least $g(\lfloor n / \delta \rfloor,h-1)$.
Summing over values of $\delta$, we count each $2$-equivalent class of $T$ as many times as its size, which concludes with the wanted bound.
\end{proof}

\subsection{The complexity of solving parity games using universal graphs}
\label{subsec:complexity_parity}

Let us start with reviewing the recent results on parity games and their relationships with separating automata, GFSG automata, and universal trees.

The first quasipolynomial time algorithm comes in three flavours: the original version~\cite{CJKLS17}, as a value iteration algorithm~\cite{FJSSW17}, and as a separating automaton~\cite{BC18}.
The second quasipolynomial time algorithm called succinct progress measure lifting algorithm~\cite{JL17} is a value iteration algorithm; it was presented using the formalism of universal trees in~\cite{Fij18}, and then as a separating automaton in~\cite{CDFJLP18}.

Hence deterministic models (separating automata) are expressive enough to capture the first two algorithms; the situation changes with the register games algorithm~\cite{Leh18}. 
Indeed it was observed in~\cite{CDFJLP18} that the algorithm induces a `non-deterministic separating automaton', meaning an automaton $\auto$ satisfying $\Omega^{|n} \subseteq L(\auto) \subseteq \Omega$; this was enough to subject this third algorithm to the lower bound on the size of separating automata since the result of~\cite{CDFJLP18} applies to non-deterministic automata (in the same way as Theorem~\ref{thm:main} does).
However this is not enough to prove the correctness of the construction: under this assumption the two games $\game$
and $\game \chain \auto$ may not be equivalent.
The follow-up paper~\cite{Parys20} investigated this issue and offered a sufficient condition for a non-deterministic separating automaton to ensure the equivalence between $\game$ and $\game \chain \auto$.
This condition is almost the same as the good-for-small-games condition we introduced independently in~\cite{CF19},
the only differences are that in the framework of~\cite{Parys20} the automata read edges (meaning triples $(v,p,v')$) 
while here the automata only read priorities, and the good-for-small-games strategy depends on the graph.
The journal version of the register games algorithm suggested the name `good-for-small-games'~\cite{LehtinenB20}.

We further discuss in the conclusions section (Section~\ref{sec:conclusions}) the recent extensions of the register games algorithm.

\vskip1em
The main conclusion of our studies on universal graphs and trees is the following algorithm,
combining Theorem~\ref{thm:universal_trees} and Theorem~\ref{thm:solving_product_game_linear}.

\begin{cor}
Let $n, m, d \in \N$.
There exists an algorithm\footnote{In the RAM model with word size $w = \log(n) \cdot \log(d)$.} solving parity games with $n$ vertices, $m$ edges, and priorities in $[0,d]$ in time 
\[
O\left(m n \cdot \binom{\lceil \log(n) \rceil + d/2 - 1}{\lceil \log(n) \rceil}\right).
\]
\end{cor}

A generous upper bound on the expression above is $n^{O(\log(d))}$.
A refined calculation reveals that the expression is polynomial in $n$ and $d$ if $d = O(\log(n))$.
We refer to~\cite{JL17} for the tightest existing analysis of the binomial coefficient.

The choice of word size $w = \log(n) \cdot \log(d)$ is inherited from the generic value iteration algorithm using word 
size $\log(|V|) + \log(|Q|)$, because $|Q| = n^{O(\log(d))}$.
For this word size it is easy to verify that the two following operations on universal trees take constant time: 
computing $\delta(q,p)$ and comparing $q \le q'$.
Choosing the word size $w = \log(n)$ implies that a leaf of the universal tree 
cannot anymore be represented in a single machine word: 
the `succinct encoding' proposed in~\cite{JL17} performs the two operations above in a very low complexity (polylogarithmic in time) in this setting.

The complexity of (variants of) this algorithm has been investigated by Chatterjee, Dvor{\'a}k, Henzinger, and Svozil 
in the set-based symbolic model~\cite{CDHS18}, which is a different computational model; the differences are only in polynomial factors.

The second outcome of our analysis is a lower bound argument.

\begin{cor}{}
\label{cor:lower_bounds_parity}
All algorithms based on separating automata, GFSG automata, or universal graphs, have at least quasipolynomial complexity.
\end{cor}

This extends the main result of~\cite{CDFJLP18} which stated this result using non-deterministic separating automata and universal trees.
The statement is kept informal since its value is not in its exact formalisation but in the perspectives it offers:
all existing quasipolynomial time algorithms for parity games have been tightly related to the notion of universal trees.
Corollary~\ref{cor:lower_bounds_parity} states that to improve further the complexity of solving parity games one needs to go beyond the notion of universal trees.

\section{Mean payoff}
\label{sec:mean_payoff}
The second objective we consider is mean payoff. 
We refer to Subsection~\ref{subsec:complexity_mean_payoff} 
for a discussion on existing results.

The set of colours is a finite set $W \subseteq \Z$ with $0 \in W$, and in this context they are called weights.
\[
\MP{W} = \set{w \in W^\omega : \liminf_n \frac{1}{n} \cdot \sum_{i = 0}^{n-1} w_i \ge 0},
\]
In this section we consider graphs over the set of colours $W$.

Let us state a characterisation of graphs satisfying $\MP{W}$ in terms of cycles.
We say that a cycle is non-negative if the total sum of weights appearing on the cycle is non-negative, and negative otherwise.
\begin{lem}
Let $G$ be a graph.
Then $G$ satisfies $\MP{W}$ if and only if all cycles in $G$ are non-negative.
\end{lem}
\begin{proof}
Assume that $G$ satisfies $\MP{W}$. A cycle induces an infinite path, and since the path satisfies $\MP{W}$ the cycle is non-negative. Conversely, if $G$ does not satisfy $\MP{W}$, a simple pumping argument implies that it contains a negative cycle.
\end{proof}

The following theorem states that the mean payoff objectives satisfy our assumptions.

\begin{thmC}[\cite{EhrenfeuchtMycielski79,GKK88}]
Mean payoff objectives are positionally determined, prefix independent, and have $0$ as neutral letter.
\end{thmC}

The goals of this section are to:
\begin{itemize}
	\item describe the structure of saturated graphs and in particular its relation with integer subsets,
	\item construct universal graphs for different mean payoff conditions,
	\item derive from these constructions efficient algorithms for solving mean payoff games,
	\item offer (asymptotically) matching lower bounds on the size of universal graphs, proving the optimality of the algorithms in this family of algorithms.
\end{itemize}

\subsection{Subsets of the integers and saturated graphs}
\label{subsec:saturated_universal_graphs_mean_payoff}

Towards understanding their structure we study the properties of saturated graphs and show that they are better presented using subsets of the integers.
A subset of integers $A \subseteq \Z$ defines a graph $\Graph(A)$ which we call an integer graph:
the set of vertices is $A$ and the set of edges is
\[
\set{(x,w,x') \in A \times W \times A : x' - x \le w}.
\]

\begin{figure}[!ht]
\centering
\includegraphics[scale=.8]{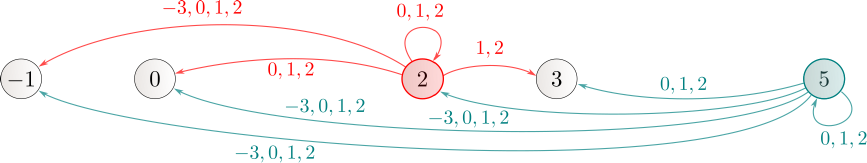}
\caption{Edges outgoing from $2$ (in blue) and $5$ (in red) in the integer graph given by subset $A=\{-1,0,2,3,5\}$ and weights $W=\{-3,0,1,2\}$. Other edges are omitted for clarity.}
\label{fig:integer_graph}
\end{figure}

Note that the definition depends on $W$.
Integer graphs contain no negative cycles and are invariant under translations: $A$ and $A + p$ for $p \in \Z$ induce the same integer graph (up to renaming vertices).

The key yet simple observation at this point is that given a graph with no negative cycles 
we can define a distance between vertices: \\
the distance $\dist(v,v')$ from a vertex $v$ to another vertex $v'$ 
is the smallest sum of the weights along a path from $v$ to~$v'$ (when such a path exists). 
Note that saturated graphs with respect to $\MP{W}$ are linear by Theorem~\ref{thm:saturated_graphs_are_linear}. 
In particular, they have a maximal element.
\begin{lem}
\label{lem:saturated_are_integer}
Let $G=(V,E)$ be a saturated graph and $v_0$ its maximal element.
Let 
\[
A = \set{\dist(v_0,v) \mid v \in V} \subseteq \N,
\]
then $G$ maps into $\Graph(A)$.
\end{lem}

\begin{proof}
We show that $\phi : V \to A$ defined by $\phi(v) = \dist(v_0,v)$ is a homomorphism from $G$ to $\Graph(A)$.
Since $v_0$ is maximal, by right-composition with the $0$-self-loop on $v_0$, for every vertex $v$ there is an edge $(v_0,0,v)$.
This implies that $\dist(v_0,v)$ is well defined, and non-negative.
Let $(v,w,v') \in E$, then $\dist(v_0,v') \le \dist(v_0,v) + w$ since a path from $v_0$ to $v$ induces a path from $v_0$ to $v'$ by adding $(v,w,v')$ at the end.
In other words, $\phi(v') - \phi(v) \le w$, so $\phi$ is indeed a homomorphism.
%
\end{proof}


The direct consequences of Lemma~\ref{lem:saturated_are_integer} is that we can restrict our attention to integer graphs.

\begin{cor}
Let $n \in \N$ and $W \subseteq \Z$.
\begin{itemize}
	\item The smallest $(n,\MP{W})$-universal graph and the smallest $(n,\MP{W})$-universal integer graph have the same size.
	\item An integer graph $\U$ is $(n,\MP{W})$-universal if and only if every integer graph of size at most $n$ maps into~$\U$.
\end{itemize}
\end{cor}
\begin{proof}
For the first item, if $\U$ is $(n,\MP{W})$-universal, without loss of generality it is saturated, 
then $\U$ maps into the integer graph constructed in Lemma~\ref{lem:saturated_are_integer}, 
so the latter is also $(n,\MP{W})$-universal.
For the second item, the direct implication is clear so we focus on the converse.
Let us assume that every integer graph of size at most $n$ maps into $\U$ and prove that $\U$ is $(n,\MP{W})$-universal.
Let $G$ be a graph of size at most $n$ satisfying $\MP{W}$, 
thanks to Lemma~\ref{lem:saturated_are_integer} it maps into an integer graph of the same size, 
which then maps into $\U$, so by composition $G$ maps into $\U$.
\end{proof}

\subsection{Upper and lower bounds for universal graphs}
\label{subsec:universal_subsets}

We now state in the following theorem upper and lower bounds on the size of universal graphs.
We present two sets of results based on $W$: the first is parametrised by the largest weight in absolute value,
and the second by the number of different weights.

\begin{thm}[Universal graphs parameterized by largest weight]\label{thm:MP_bounds_abs}
Let $n,N \in \N$, and $W=[-N,N]$.
\begin{itemize}
	\item There exist $(n,\MP{W})$-universal graphs of size $2n \cdot (nN)^{1 - 1/n}$ and of size $(n-1)N+1$.
 	\item All $(n,\MP{W})$-universal graphs have size at least $(N+1)^{1 - 1/n}$.
\end{itemize}
\end{thm}

\begin{thm}[Universal graphs parameterised by the number of weights]\label{thm:MP_bounds_card}
\hfill
\begin{itemize}
	\item For all $n \in \N$ and $W$ of cardinality $k$, there exists an $(n,\MP{W})$-universal graph of size $O(n^k)$.
	\item For all $k$ and for all $n$ large enough, there exists $W \subseteq \Z$ of cardinality $k$ such that all $(n,\MP{W})$-universal graphs have size $\Omega(n^{k-2})$.
\end{itemize}
\end{thm}

Let us start with simple upper bounds.

\begin{lem}
\label{lem:upper_bounds_universal_graphs_mean_payoff}
Let $n,N \in \N$.
\begin{itemize}
	\item The integer graph $\Graph([0,(n-1)N)])$ is $(n,\MP{[-N,N]})$-universal.
	\item For every $W$ of cardinality $k$, there exists an $(n,\MP{W})$-universal graph of size $(n-1)^k$.
\end{itemize}
\end{lem}

\begin{proof}
We show both results at the same time. 

Let $G$ be a graph of size at most $n$ satisfying $\MP{W}$, 
and $\widehat{G}$ a saturation of $G$, thanks to Lemma~\ref{lem:saturated_are_integer}
$\widehat{G}$ maps into $\Graph(A)$, implying that $G$ as well.

In the case where $W = [-N,N]$ we have $A \subseteq [0,(n-1)N]$ so $\Graph(A)$ maps into $G([0,(n-1)N])$.
This proves that the integer graph $\Graph([0,(n-1)N])$ is $(n,\MP{[-N,N]})$-universal.

In the case where $|W| \leq k$ we consider the set $\Sigma_{W,n-1}$ of all sums of $\leq n-1$ weights of $W$, 
which has cardinality at most $(n-1)^k$. 
Then $A \subseteq \Sigma_{W,n-1}$ so $\Graph(A)$ maps into $\Graph(\Sigma_{W,n-1})$.
This proves that $\Graph(\Sigma_{W,n-1})$ is $(n,\MP{W})$-universal and it has size $(n-1)^k$.
\end{proof}

We state here a simple result that we will use several times later on about homomorphisms into integer graphs.
\begin{fact}
\label{fact:equality}
Let $G$ be a graph and $\phi : G \to \Graph(A)$ a homomorphism into an integer graph. 
Let us consider a cycle
$v_0w_0v_1 \cdots v_{\ell-1}w_{\ell-1}v_0$ in $G$ of total weight $0$.
Then for $i \in [0,\ell)$, we have 
$\phi(v_{i+1}) - \phi(v_i) = w_i$,
where by convention $v_\ell = v_0$.
\end{fact}

\begin{proof}
By definition of the homomorphism using the edges $(v_i,w_i,v_{i+1})$ for $i \in [0,\ell)$
\[
\begin{array}{lll}
\phi(v_1) - \phi(v_0) & \le & w_0 \\
\phi(v_2) - \phi(v_1) & \le & w_1 \\
& \ \vdots & \\
\phi(v_0) - \phi(v_{\ell-1}) & \le & w_{\ell-1}.
\end{array}
\]
Assume towards contradiction that one of these inequalities is strict. 
Summing all of them yields $0$ on both sides (because the cycle has total weight $0$), so this would imply $0 < 0$, a contradiction.
Hence all inequalities are indeed equalities.
\end{proof}

We further analyse universal graphs when $W = [-N,N]$.
We already explained in Lemma~\ref{lem:upper_bounds_universal_graphs_mean_payoff} 
how to construct an $(n,\MP{[-N,N]})$-universal graphs of size $(n-1)N+1$.
We improve on this upper bound when $N$ is exponential in $n$.

\begin{prop}
\label{prop:upper_bound_largest_weight}
There exists an $(n,\MP{[-N,N]})$-universal graph of size 
\[
2 \Big((n-1)N - \big[((n-1)N)^{1/n} - 1 \big]^n \Big) \leq 2n \cdot (nN)^{1-1/n}.
\]
\end{prop}

As discussed above, the size of this new universal graph is not always smaller than the first universal graph of size $(n-1)N+1$, 
but it is asymptotically smaller when $n^n = o(N)$.
We now give some intuition for the construction.
The first part of the proof of Lemma~\ref{lem:upper_bounds_universal_graphs_mean_payoff} shows that it suffices to embed integer graphs of the form $\Graph(A)$, when $A = \set{\dist(v_0,v), v \in V}$ for some graph $G$.
In the above proof, this was done by finding a large enough set of integers that includes any such $A$:
in the first case $B = [0,(n-1)N]$ and in the second case $B = \Sigma_{W,n-1}$. 
To conclude we used the fact that if $A \subseteq B$ then there exists a homomorphism from $\Graph(A)$ to $\Graph(B)$.

A weaker condition is the existence of $p \in \Z$ such that $A + p \subseteq B$: 
this still implies the existence of a homomorphism from $A$ to $B$.
This will allow us to remove some values from $[0,(n-1)N]$ while remaining universal. As a drawback we need to double the range to $[0,2(n-1)N)$, which explains why this new construction is not always smaller than the first one.

\begin{proof}
Let $b = ((n-1)N)^{1/n}$.
We write integers $a \in [0,2(n-1)N)$ in basis $b$, hence using $n+1$ digits written $a[i] \in [0,b)$, that is,
\[
a = \sum_{i = 0}^n a[i] b^i = \sum_{i = 0}^n a[i] ((n-1)N)^{i/n}.
\]
Note that since $a \in [0,2(n-1)N)$ the $(n+1)$\textsuperscript{th} digit is either $0$ or $1$.
We let $B$ be the set of integers in $[0,2(n-1)N)$ which have at least one zero digit among the first $n$ digits in this decomposition.
We argue that $\Graph(B)$ is $(n,\MP{[-N,N]})$-universal.

Let $G=(V,E)$ be a graph of size at most $n$ satisfying $\MP{W}$, to show that 
$G$ maps into $\Graph(B)$ we use Lemma~\ref{lem:saturated_are_integer}
and show that for $A = \set{\dist(v_0,v) : v \in V}$, the graph $\Graph(A)$ maps into $\Graph(B)$.
To this end we show that there exists $p \in \Z$ such that $p + A \subseteq B$.

Let $A = \{d_0 < \cdots < d_{n-1}\} \subseteq \Z$, where $d_0 = \dist(v_0,v_0)=0$, and for all $i \in [0,n-2]$, we have $1 \leq d_{i+1}- d_i \leq N$.

We choose $p$ of the form $\sum_{i = 0}^{n-1} a_i b^i$ for $a_i \in [0,b)$,
\textit{i.e.} the $(n+1)$\textsuperscript{th} digit is $0$, or equivalently $p < (n-1)N$.
Let us write $p_i = p + d_i$, for $i \in [0,n-1]$. 
We need to choose $p$ such that for every $i \in [0,n-1]$ we have $p_i \in B$.
Note that for all $i \in [0,n-1]$ it holds that $d_i \in [0,(n-1)N)$, so whatever the choice of $p \in [0,(n-1)N]$, we have $p_i \in [0,2(n-1)N)$.

We show how to choose $a_0, a_1, \dots, a_{n-1}$ in order to ensure that $p_0,p_1,\dots, p_{n-1} \in B$, that is, each have at least one digit among the first $n$ ones which is zero.
More precisely, we show by induction on $k \in [0,n-1]$ that there exist $a_0,\dots, a_k \in [0,b)$ such that for any choice of $a_{k+1}, \dots, a_{n-1} \in [0,b)$ and for all $i \in [0,k]$, the $i$-th digit $p_i[i]$ of $p_i$ is $0$. 
For $k = 0$, we let $a_0 = 0$, which yields $p_0[0] = 0$, independently of the values of $a_1, \dots, a_{n-1}$.

Let $a_0,\dots, a_{k-1}$ be such that for any choice of $a_k, \dots, a_{n-1}$, 
for any $i \in [0,k-1]$, we have $p_i[i] = 0$. 
Let $a_k \in [0,b)$ be the (unique) value such that 
$\left(\sum_{i = 0}^k a_i b^i + d_k \right)[k] = 0$. 
Let $a_{k+1}, \dots, a_{n-1} \in [0,b)$. 
By induction hypothesis, for any $i \in [0,k)$, we have $p_i[i] = 0$. 
Now
\[
p_k[k] = \left(p + d_k \right)[k]
 = \left(\sum_{i = 0}^{n-1} a_i b^i + d_k \right)[k]
 = \left(\sum_{i = 0}^{k} a_i b^i + d_k \right)[k] + \left(b^{k+1} \sum_{i = k+1}^{n-1} a_i b^{i-k}\right)[k] = 0,
\]
since both terms are zero.
This concludes the inductive construction of $p$, and the proof of universality of $B$.
Since $B$ excludes from $[0,2(n-1)N)$ exactly $2(b-1)^n$ integers (those that do not use the digit $0$ in their first $n$ digits), the size of $B$ is
\[
2 \Big((n-1)N - \big[((n-1)N)^{1/n} - 1 \big]^n \Big). \qedhere
\]
\end{proof}

We now prove the lower bound of Theorem~\ref{thm:MP_bounds_abs}.

\begin{prop}
\label{prop:lower_bound_largest_weight}
Any $(n,\MP{[-N,N]})$-universal graph has size at least $(N+1)^{1 - 1/n}$.
\end{prop}

\begin{proof}
Thanks to Lemma~\ref{lem:saturated_are_integer}, we let $\U = \Graph(U)$ be an $(n,\MP{[-N,N]})$-universal integer-graph given by $U \subseteq Z$. 
We construct an injective function
$f : [0,N]^{n-1} \to U^n$.
For $(w_1,\dots,w_{n-1}) \in [0,N]^{n-1}$, we consider 
$A = \{a_0,a_1, \dots a_{n-1}\} \subseteq \Z$, with $a_0=0$ and $a_j = \sum_{i=0}^j w_j$ for $j \geq 1$.

By universality, $\Graph(A)$ maps into $\U$, and we let $\phi$ be such a homomorphism.
Define $f(w_1,\dots,w_{n-1})$ to be the $n$-tuple $(\phi(a_0), \phi(a_1), \dots, \phi(a_{n-1}))$ of integer in $U$ given by
\[
f(w_1, \dots, w_{n-1}) = (\phi(a_0),\phi(a_1), \dots,\phi(a_{n-1})).
\]
To see that $f$ is injective, we apply Fact~\ref{fact:equality} to each cycle of length 2 in $\Graph(A)$ of the form
\[
a_{j},w_{j+1},a_{j+1},-w_j,a_j,
\]
for $j \in [0,n-2]$. This yields
\[
\begin{array}{lll}
\phi(a_1) - \phi(a_0) & = & w_1 \\
\phi(a_2) - \phi(a_1) & = & w_2 \\
& \ \vdots & \\
\phi (a_{n-1}) - \phi(a_{n-2}) & = & w_{n-1},
\end{array}
\]
and $f$ is injective.
We conclude that 
$(N+1)^{n-1} \le |\U|^n$, so $|\U| \ge (N+1)^{1 - 1/n}$.
\end{proof}

We finally present the lower bound of Theorem~\ref{thm:MP_bounds_card}.
We let 
$T = 1 + n + n^2 + \cdots + n^{k-2}$
and 
\[
W = \set{1,n,n^2,\dots,n^{k-2}, - \frac{n-1}{k-1} T}.
\]
Note that $W$ has indeed cardinality $k$.

\begin{prop}
\label{prop:lower_bound_number_weights}
Let $\U$ be an $(n,\MP{W})$-universal graph.
Then 
\[
|\U| \geq \left( \frac{n-1}{(k-1)^2} \right)^{\frac{(k-1)^2}{k}}.
\]
\end{prop}

\begin{proof}
Let $\U$ be an $(n,\MP{W})$-universal graph.

We consider a class of graphs which are cycles of length $n$, and later use a subset of those for the lower bound.
Let $(w_1,\dots,w_{n-1}) \in \set{1,n,\dots,n^{k-2}}$.
The vertices are $[0,n)$. 
There is an edge $(i-1,w_i,i)$ for $i \in [1,n)$ and an edge $(n-1, -\frac{n-1}{k-1} T, 0)$.
To make the total weight in the cycle equal to $0$,
we assume that each $n^j$ appears exactly $\frac{n-1}{k-1}$ many times in $w_1,\dots,w_{n-1}$.

The cycles we use for the lower bound are described in the following way.
We let $S$ be the set of sequences of $k$ integers in $[0,n)$ such that $\sum_{\ell = 1}^k s_\ell = \frac{n-1}{k-1}$,
where we use the notation $(s_1, s_2, \dots, s_k)$ for an element~$s \in S$.
A tuple of $k-1$ sequences in $S$ induces a graph $G$.
Let $(s^{(0)}, \dots, s^{(k-2)}) \in S^{k-1}$, the induced graph is partitioned into $k$ parts.
In the $i$\textsuperscript{th} part the weight $n^j$ is used exactly $s^{(j)}_i$ many times. See Figure~\ref{fig:construction}.

\begin{figure}[ht]
\centering
\includegraphics[width=.7\linewidth]{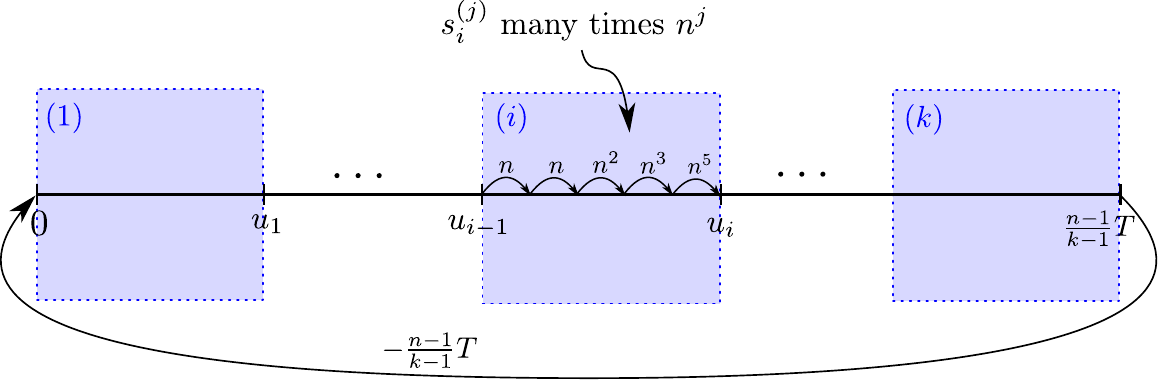}
\caption{Construction of the graph $G$ from the sequences $(s^{(0)}, \dots, s^{(k-2)}) \in S^{k-1}$. In this example, the $i$\textsuperscript{th} box is induced by the $i$\textsuperscript{th} sequence $s^{(i)}$ given by $s_0^{(i)} = 0, s^{(i)}_1 = 2, s^{(i)}_2 = 1, s^{(i)}_3 = 1, s^{(i)}_4 = 0$, and $s^{(i)}_5 = 1$.}
\label{fig:construction}
\end{figure}

The vertex in $G$ marking the end of the first box is $u_1 = \sum_{j = 0}^{k-2} s_1^{(j)} n^j$,
and more generally the vertex marking the end of the $i$\textsuperscript{th} box is 
\[
u_i = \sum_{j = 0}^{k-2} \left( \sum_{\ell = 1}^i s_\ell^{(j)} \right) n^j.
\]
Note that $\sum_{\ell = 1}^i s_\ell^{(j)}$ is the $j$\textsuperscript{th} digit of the number $u_i$ in base $n$,
since it belongs to $[0,\frac{n-1}{k-1}]$, and so it is $<n$.
Hence, $u_1,\dots,u_{k-1}$ fully determines the sequences $(s^{(0)}, \dots, s^{(k-2)})$.
We also let $u_0 = 0$.

Let $\phi : G \to \U$ be a homomorphism.
We define a function $f : S^{k-1} \to \U^k$ by
$f(s^{(0)}, \dots, s^{(k-2)}) = \left( \phi \left( u_i \right) : i \in [0,k) \right)$, which we prove to be injective.
The constraint on the sums of sequences on $S$ ensures that indeed the graph $G$ is a cycle of total weight~$0$, and then
thanks to Lemma~\ref{fact:equality} applied to the cycle $G$, we have
\[
\begin{array}{lll}
\phi(u_1) - \phi(u_0) & = & u_1 \\
\phi(u_2) - \phi(u_1) & = & u_2 \\
& \ \vdots & \\
\phi(u_{k-1}) - \phi(u_{k-2}) & = & u_{k-1},
\end{array}
\]
and as explained above the numbers $u_1,\dots,u_{k-1}$ fully determine the sequences $(s^{(0)}, \dots, s^{(k-2)})$.

Injectivity of $f$ implies that $|S|^{k-1} \le |\U|^k$.
The size of $S$ is 
\[
|S| = \binom{\frac{n-1}{k-1} + k-1}{k-1}
\ge \left( \frac{\frac{n-1}{k-1} + k-1}{k-1} \right)^{k-1} 
\ge \left( \frac{n-1}{(k-1)^2} \right)^{k-1},
\]
which implies (for $k$ constant)
$|\U| \geq \Omega \left( n^{ \frac{(k-1)^2}{k} } \right) = \Omega \left( n^{k-2} \right)$.
\end{proof}

\subsection{The complexity of solving mean payoff games using universal graphs}
\label{subsec:complexity_mean_payoff}

There is a very large literature on mean payoff games, a model independently introduced 
by Ehrenfeucht and Mycielski~\cite{EhrenfeuchtMycielski79} and by Gurvich, Karzanov, and Khachiyan~\cite{GKK88}
and studied in different research communities including program verification and optimisation.
The seminal paper of Zwick and Paterson~\cite{ZwickPaterson96} relates mean payoff games to discounted payoff games
and simple stochastic games, and most relevant to our work, constructs an algorithm for solving mean payoff games with complexity $O(n^2 m N)$,
where $n$ is the number of vertices, $m$ the number of edges, and $N$ the largest weight in absolute value.
If the weights are given in unary $N$ is polynomial in the representation, so we say that the algorithm is pseudo-polynomial.
The question whether there exists a polynomial time algorithm for mean payoff games with the usual representation of weights, meaning in binary,
is open.
The currently fastest algorithm for mean payoff games is randomised and achieves subexponential complexity $2^{\tilde{O}(\sqrt{n})}$.
It is based on randomised pivoting rules for the simplex algorithm devised by Kalai~\cite{Kalai92,Kalai97} and Matou\v{s}ek, Sharir and Welzl~\cite{MSW96}.

We are in this work interested in deterministic algorithms for solving mean payoff games.
Up until recently there were two fastest deterministic algorithms: the value iteration algorithm of Brim, Chaloupka, Doyen, Gentilini, and Raskin~\cite{BCDGR11}, which has complexity $O(n m N)$, and the algorithm of Lifshits and Pavlov~\cite{LP07} with complexity $O(n m 2^n)$.
They are incomparable: the former is better when $N \le 2^n$ and otherwise the latter prevails.
Recently Dorfman, Kaplan, and Zwick~\cite{DKZ19} presented an improved version of the value iteration algorithm
with a complexity $O(\min (n m N, n m 2^{n/2}))$, an improvement over the previous two algorithms when $N = \Omega(n 2^{n/2})$.

Solving a mean payoff game is very related to constructing an optimal strategy, meaning one achieving the highest possible value.
The state of the art for this problem is due to Comin and Rizzi~\cite{CominRizzi17} who designed a pseudo-polynomial time algorithm.

\vskip1em
Let us now construct algorithms from universal graphs combining Theorem~\ref{thm:MP_bounds_abs} and Theorem~\ref{thm:MP_bounds_card} with Theorem~\ref{thm:solving_product_game_linear}.

\begin{cor}
Let $n, N \in \N$.
The following statements use the RAM model with word size $w = \log(n) + \log(N)$.
\begin{itemize}
	\item There exists an algorithm for solving mean payoff games with weights in $[-N,N]$ of time complexity $O(nmN)$ and space complexity $O(n)$.
	\item There exists an algorithm for solving mean payoff games with weights in $[-N,N]$ of time complexity 
	$O(mn (nN)^{1 - 1/n})$ and with space complexity $O(n)$.
	\item There exists an algorithm for solving mean payoff games with $k$ weights of time complexity $O(m n^k)$ and space complexity $O(n)$.
\end{itemize}
\end{cor} 

The first algorithm is exactly the algorithm constructed by Brim, Chaloupka, Doyen, Gentilini, and Raskin~\cite{BCDGR11}: identical data structures and complexity analysis. 
The two tasks for manipulating universal graphs, 
namely computing $\delta(q,w)$ and checking whether $q \le q'$ 
are indeed unitary operations as they manipulate numbers of order $nN$.

Let us now discuss the significance of the difference between $N$ to $N^{1 - 1/n}$,
meaning the first and second algorithms.
Since $n$ is polynomial in the size of the input, one may say that $n$ is ``small'', 
while $N$ is exponential in the size of the input when weights are given in binary, hence ``large''.
For $N \le 2^n$ we have $N^{1 - 1/n} \ge \frac{1}{2} N$ so $N^{1 - 1/n}$ is essentially linear in $N$.
However when $N \ge 2^{\Omega(n^{1 + \varepsilon})}$ for $\varepsilon > 0$ then $2n^2 N^{1 - 1/n} = o(nN)$, 
so the second algorithm is indeed asymptotically faster.
To appreciate the relevance of this condition, let us recall an old result of Frank and Tardos~\cite{FT87} which implies that one can in polynomial time transform a mean payoff game into an equivalent one where $N \le 2^{4n^3} m^{m + 3}$.
Our new algorithm improves over the previous one for the range $N \in [2^{\Omega(n^{1 + \varepsilon})}, 2^{4n^3} m^{m + 3}]$.
We note however that the recent algorithm~\cite{DKZ19} (building over~\cite{BCDGR11} and~\cite{GKK88}) is always faster than the second algorithm.

\vskip1em
As for the case of parity games, the lower bounds presented in Theorem~\ref{thm:MP_bounds_abs} and Theorem~\ref{thm:MP_bounds_card} show that to obtain faster algorithms we need to go beyond algorithms
based on universal graphs.
In particular there are no quasipolynomial universal graphs for mean payoff objectives, 
and the question whether the quasipolynomial time algorithms for parity games can be extended to mean payoff games
remains open.

\section{Disjunction of a parity and a mean payoff objective}
\label{sec:disj_mean_payoff_parity}
The third objective is the disjunction of parity and mean payoff objectives.
We refer to Subsection~\ref{subsec:complexity_disjunctions_parity_mean_payoff} 
for a discussion on existing results.

The set of colours is $[0,d] \times W$ for some $W \subseteq \Z$ with $0 \in W$.
For $w \in ([0,d] \times W)^\omega$ we write $w_P \in [0,d]^\omega$ for the projection on the first component
and $w_{MP} \in W^\omega$ for the projection on the second component.
We define $\Parity(d) \vee \MP{W}$ as
\[
\set{w \in ([0,d] \times W)^\omega : w_P \in \Parity(d) \vee w_{MP} \in \MP{W}}.
\]

The following theorem states that disjunctions of parity and mean payoff objectives satisfy our assumptions.

\begin{thmC}[\cite{ChatterjeeHJ05}]
Disjunctions of parity and mean payoff objectives are positionally determined, prefix independent, and have $(0,0)$ 
as neutral letter.
\end{thmC}

Our approach for constructing universal graphs for disjunctions of parity and mean payoff objectives
differs from the previous two cases.
Here we will rely on the existing constructions of universal graphs for parity and for mean payoff objectives, 
and show general principles for combining universal graphs for disjunctions of objectives.

It is here technically more convenient to reason with separating automata, taking advantage of determinism. 
Thanks to Theorem~\ref{thm:main} this is equivalent in terms of complexity of the obtained algorithms.

\subsection{Separating automata for disjunctions of parity and mean payoff objectives} 
\label{subsec:universal_graphs_disjunctions_parity_mean_payoff}

\begin{thm}
\label{thm:separating_automata_disj_parity_mp}
Let $n,d \in \N, W \subseteq \Z$, $\auto_P$ an $(n,\Parity(d))$-separating automaton, and $\auto_{MP}$ an $(n,\MP{W})$-separating automaton.
Then there exists an $(n, \Parity(d) \vee \MP{W})$-separating automaton of size $O(d \cdot |\auto_P| \cdot |\auto_{MP}|)$.
\end{thm}

\begin{proof}
Let us write $\auto_P = (Q_P,q_{0,P},\delta_P)$ and $\auto_{MP} = (Q_{MP},q_{0,MP},\delta_{MP})$.
We define a deterministic automaton $\auto_{P \vee MP}$: 
the set of states is $[0,d] \times Q_P \times Q_{MP}$,
the initial state is $(d, q_{0,P}, q_{0,MP})$,
and the transition function is defined as follows:
$\delta((p,q_P,q_{MP}), (p',w))$ is equal to
\[
\begin{cases}
(\max(p,p'), q_P, q'_{MP}) & \text{ if $\delta_{MP}(q_{MP}, w)$ is defined and equal to $q'_{MP}$,} \\
(0, q'_P, q_{0,MP}) & \text{ otherwise, and $\delta_P(q_{P}, \max(p,p'))$ is defined and equal to $q'_{P}$}.
\end{cases}
\]
Intuitively: $\auto_{P \vee MP}$ simulates the automaton $\auto_{MP}$, storing the maximal priority seen since the last reset (or from the beginning).
If the automaton $\auto_{MP}$ rejects, the automaton resets, which means two things:
it simulates one transition of the automaton $\auto_{P}$ using the stored priority, 
and resets the state of $\auto_{MP}$ to its initial state.
The automaton $\auto_{P \vee MP}$ rejects only if $\auto_{P}$ rejects during a reset.

We now prove that $\auto_{P \vee MP}$ is an $(n,\Parity(d) \vee \MP{W})$-separating automaton.
\begin{itemize}
	\item $L(\auto_{P \vee MP}) \subseteq \Parity(d) \vee \MP{W}$.
	Let $\pi$ be an infinite path accepted by $\auto_{P \vee MP}$, we distinguish two cases by looking at the run of $\pi$.

	\begin{itemize}
		\item If the run is reset finitely many times, let us write $\pi = \pi' \pi''$
		where $\pi'$ is finite and $\pi'_{MP}$ rejected by $\auto_{MP}$, and $\pi''$ is infinite and $\pi''_{MP}$ is accepted by $\auto_{MP}$.
		Since $\auto_{MP}$ satisfies $\MP{W}$, this implies that $\pi''_{MP}$ satisfies $\MP{W}$,
		so by prefix independence, $\pi_{MP}$ satisfies $\MP{W}$ hence $\pi$ satisfies $\Parity(d) \vee \MP{W}$.
	
		\item If the run is reset infinitely many times, we may find an infinite decomposition $\pi = \pi^1 \pi^2 \dots$
		where for each $i \in \N$, the path $\pi^i_{MP}$ is rejected by $\auto_{MP}$, and its proper prefixes are not.
		Let $p_i$ be the maximum priority appearing in $\pi^i_{P}$, the run of $\pi$ over $\auto_{P \vee MP}$
		induces a run of $p_1 p_2 \dots$ over $\auto_P$.
		Since $\auto_P$ satisfies $\Parity(d)$, this implies that $p_1 p_2 \dots$ satisfies $\Parity(d)$.
		By definition of the $p_i$'s, this implies that $\pi_P$ satisfies $\Parity(d)$ so a fortiori $\pi$ satisfies $\Parity(d) \vee \MP{W}$.
	\end{itemize}		
	
	\item $(\Parity(d) \vee \MP{W})^{\mid n} \subseteq L(\auto_{P \vee MP})$.
	Let $G=(V,E)$ be a graph satisfying $\Parity(d) \vee \MP{W}$ of size at most $n$,
	we need to show that all infinite paths of $G$ are accepted by $\auto_{P \vee MP}$.
	
	We construct a graph $G_{MP}=(V, E_{MP})$ over the set of colours $W$ and a graph $G_P=(V, E_P)$ over the set of colours $[0,d]$.
	Both use the same set of vertices $V$ as $G$. 
	We prove that the graph $G_{MP}$ satisfies $\MP{W}$, which is used to prove that $G_P$ satisfies $\Parity(d)$.
	\begin{itemize}
		\item \textbf{The graph $G_{MP}$.} 
	There is an edge $(v,w,v') \in E_{MP}$ if 
	there exists $p \in [0,d]$ such that $e = (v,(p,w),v') \in E$,
	and either $p$ is odd or $p$ is even and $e$ is not contained in any negative cycle with maximum priority $p$.
	
	\vskip1em	
	We claim that $G_{MP}$ satisfies $\MP{W}$.
	Assume towards contradiction that $G_{MP}$ contains a negative cycle $C_{MP}$. 
	It induces a negative cycle $C$ in $G$, by definition of $G_{MP}$ necessarily the maximum priority in $C$ is odd.
	Hence $C$ is a negative odd cycle in $G$, a contradiction.
	
		\item \textbf{The graph $G_{P}$.}
	There is an edge $(v,p,v') \in E_P$ if 
	there exists a path in $G$ from $v$ to $v'$ 
	with maximum priority $p$ and (whose projection on the mean payoff component is) rejected by $\auto_{MP}$.

	\vskip1em	
	We claim that $G_{P}$ satisfies $\Parity(d)$.
	Assume towards contradiction that $G_P$ contains an odd cycle $C_P$.
	For each edge in this cycle there is a corresponding path rejected by $\auto_{MP}$
	with the same maximum priority.
	Putting these paths together yields an odd cycle $C$, of maximal priority $p$, in $G$ whose projection on the mean payoff component 
	is rejected by $\auto_{MP}$.
	Since $\MP{W}^{\mid n} \subseteq L(\auto_{MP})$ and as we have shown, $G_{MP}$ satisfies $\MP{W}$, 
	the projection of $C$ on the mean payoff component is not in $G_{MP}$,
	so there exists an edge $(v,(p',w),v')$ in $C$ such that $(v,w,v')$ is not in $E_{MP}$. 
	This implies that $p'$ is even, so in particular $p' < p$, and $(v,(p',w),v')$ is contained in a negative cycle $C'$ in $G$ with maximum priority $p'$.	
	Combining the odd cycle $C$ followed by sufficiently many iterations of the negative cycle $C'$ yields a path in $G$, with negative weight, and maximal priority $p$ which is odd, a contradiction.
	\end{itemize}
			
	Let $\pi$ an infinite path in $G$, we show that $\pi$ is accepted by $\auto_{P \vee MP}$.
	Let us consider first only the mean payoff component: we run $\pi$ repeatedly 	over $\auto_{MP}$,
	and distinguish two cases.
	
	\begin{itemize}
		\item If there are finitely many resets, let us write $\pi = \pi_1 \pi_2 \dots \pi_k \pi'$
		where $\pi_1,\dots,\pi_k$ are paths rejected by $\auto_{MP}$ with proper prefixes accepted by $\auto_{MP}$ and $\pi'$ is accepted by $\auto_{MP}$.
		To show that $\pi$ is accepted by $\auto_{P \vee MP}$ we need to show that the automaton $\auto_P$
		accepts the word $p_1 \dots p_k$ where $p_i$ is the maximum priority appearing in $\pi_i$ for $i \in [1,k]$.
		Indeed, $p_1 \dots p_k$ is a path in $G_P$, which is a graph of $n$ satisfying $\Parity(d)$,
		so $\auto_P$ accepts $p_1 \dots p_k$.
	
		\item If there are infinitely many resets, let us write $\pi = \pi_1 \pi_2 \dots$
		where for each $i \in \N$, the path $\pi_i$ is rejected by $\auto_{MP}$ and its proper prefixes are not.
		To show that $\pi$ is accepted by $\auto_{P \vee MP}$ we need to show that the automaton $\auto_P$
		accepts the word $p_1 p_2 \dots$, which holds for the same reason as the other case: 
		$p_1 p_2 \dots$ is a path in $G_P$. \qedhere
	\end{itemize}		
\end{itemize}

\end{proof}

\subsection{The complexity of solving disjunctions of parity and mean payoff games using separating automata}
\label{subsec:complexity_disjunctions_parity_mean_payoff}

The class of games with a disjunction of a parity and a mean payoff objective has been introduced in~\cite{ChatterjeeHJ05},
with a twist: this paper studies the case of a conjunction instead of a disjunction, which is more natural in many applications.
This is equivalent here since both parity and mean payoff objectives are dual, in other words
if the objective of Eve is a disjunction of a parity and a mean payoff objective, then 
the objective of Adam is a conjunction of a parity and a mean payoff objective.
Hence all existing results apply here with suitable changes.
The reason why we consider the objective of the opponent is that it is positionally determined,
which is the key assumption for using the universal graph technology.

The state of the art for solving disjunctions of parity and mean payoff games is due to~\cite{DaviaudJL18},
which presents a pseudo-quasi-polynomial algorithm.
We refer to~\cite{DaviaudJL18} for references on previous studies for this class of games.
As they explain, these games are logarithmic space equivalent to the same games replacing mean payoff by energy,
and polynomial time equivalent to games with weights~\cite{ScheweWZ19}, extending games with costs~\cite{FM14}.

Combining Theorem~\ref{thm:separating_automata_disj_parity_mp} with Theorem~\ref{thm:solving_product_game_linear}
yields the following result.

\begin{thm}
Let $n,d,N \in \N$.
There exists an algorithm~\footnote{In the RAM model with word size $w = \log(n) + \log(N)$.} 
for solving disjunctions of parity and mean payoff games
with $n$ vertices, $m$ edges, weights in $(-N,N)$ and priorities in $[0,d]$
of time complexity 
\[
O\left(m d \cdot \underbrace{n \cdot \binom{\lceil \log(n) \rceil + d/2 - 1}{\lceil \log(n) \rceil}}_{\text{Parity}}
\cdot \underbrace{nN}_{\text{Mean Payoff}}
\right).
\]
and space complexity $O(n)$.
\end{thm}

Our algorithm is similar to the one constructed in~\cite{DaviaudJL18}: they are both value iteration algorithms (called progress measure lifting algorithm in~\cite{DaviaudJL18}), combining the two value iteration algorithms for parity and mean payoff games.
However, the set of values are not the same (our algorithm stores an additional priority) and the proofs are very different. 
Besides being much shorter\footnote{The only generic result we use from universal graphs here is the equivalence between $\game$ and $\game \times \auto$ when $\auto$ is a separating automaton.}, one advantage of our proof is that it works with abstract universal graphs for both parity and mean payoff objectives, and shows how to combine them, whereas in~\cite{DaviaudJL18} the proof is done from scratch, extending both proofs for the parity and mean payoff objectives.

\section{Disjunction of mean payoff objectives}
\label{sec:disj_mean_payoff}
The fourth objective we consider is disjunctions of mean payoff objectives.
We refer to Subsection~\ref{subsec:complexity_disjunctions_mean_payoff} 
for a discussion on existing results.

The set of colours is a finite set $W^d \subseteq \Z^d$ for some $W \subseteq \Z$ with $0 \in W$.
For $w \in (W^d)^\omega$ and $i \in [1,d]$ we write $w_i \in W^\omega$ for the projection on the $i$-th component.
\[
\bigvee_{i \in [1,d]} \MP{W}_i = \set{w \in (W^d)^\omega : \exists i \in [1,d], w_i \in \MP{W}}.
\]

The following theorem states that disjunctions of mean payoff objectives satisfy our assumptions.

\begin{thm}[Lemma 9 in~\cite{VelnerC0HRR15}]
Disjunctions of mean payoff objectives are positionally determined, prefix independent, and have $\overline{0}$ 
as neutral letter.
\end{thm}

To obtain a universal graph for disjunctions of mean payoff objectives, we will show how to combine the existing constructions for mean payoff objectives.

\subsection{A general reduction of universal graphs for strongly connected graphs} 
\label{subsec:reduction_strongly_connected}
The first idea is very general, it roughly says that if we know how to construct universal graphs for strongly connected graphs,
then we can use them to construct universal graphs for general graphs.

We say that a graph $G$ is strongly connected if for every pair of vertices there exists a path from one vertex to the other.
Let us refine the notion of universal graphs.

\begin{defi}[Universal graphs for strongly connected graphs]
A graph $\U$ is $(n,\Omega)$-universal for strongly connected graphs if it satisfies $\Omega$ and all strongly connected graphs of size at most $n$ satisfying $\Omega$ map into $\U$.
\end{defi}

Let us give a first construction, which we refine later.
Let $G_1=(V_1,E_1)$ and $G_2=(V_2,E_2)$ be two graphs, we define their sequential product $\langle G_1,G_2 \rangle$ as follows:
the set of vertices is $V_1 \cup V_2$, and the set of edges is
\[
E_1 \ \cup\ E_2 \ \cup\ \left(V_1 \times C \times V_2\right).
\]
The sequential product is extended inductively:
$\langle G_1,\dots,G_p \rangle = \langle \langle G_1,\dots,G_{p-1} \rangle, G_p \rangle$.
Intuitively, $\langle G_1,\dots,G_p \rangle$ is the disjoint union of the graphs $G_i$'s with all possible forward edges
from $G_i$ to $G_j$ for any $i < j$.

\begin{lem}
\label{lem:strongly_connected_to_general_naive}
Let $\Omega$ be a prefix independent objective, $n \in \N$, 
and $\U$ an $(n,\Omega)$-universal graph for strongly connected graphs.
Then $\U^n = \langle \underbrace{\U,\dots,\U}_{n \text{ times}} \rangle$ is an $(n,\Omega)$-universal graph.
\end{lem}

\begin{proof}
We first note that any infinite path in $\U^n$ eventually remains in one copy of $\U$. Since $\U$ satisfies $\Omega$ and by prefix independence of $\Omega$, this implies that the infinite path satisfies $\Omega$, thus so does $\U^n$.

Let $G=(V,E)$ be a graph satisfying $\Omega$, we show that $G$ maps into $\U$.
We decompose $G$ into strongly connected components: 
let $G_1=(V_1,E_1),\dots,G_p=(V_p,E_p)$ be the maximal strongly connected components in $G$
indexed such that if there exists an edge from $G_i$ to $G_j$ then $i < j$.
Then $V$ is the disjoint union of the $V_i$'s.
Each $G_i$ is a subgraph of $G$, and since $G$ satisfies $\Omega$ then so does $G_i$.
It follows that each $G_i$ maps into $\U$, let $\phi_i$ be a homomorphism $\phi_i : G_i \to \U$.

Let us define $\phi : G \to \U^n$ by $\phi(v) = \phi_i(v)$ if $v \in V_i$.
We now argue that $\phi$ is a homomorphism $\phi : G \to \U^n$.
Let $e = (v,c,v') \in E$, 
then either $e \in E_i$, and then $\phi(e)$ belongs to $\U^n$
because $\phi_i$ is a homomorphism,
or $e$ goes from $G_i$ to $G_j$ with $i < j$, and then again $\phi(e)$ is an edge in $\U^n$.
Thus $\phi$ is a homomorphism and therefore $\U^n$ is $(n,\Omega)$-universal.
\end{proof}

In the construction above we have used the fact that $G$ decomposes into at most $n$ strongly connected components of size $n$.
We refine this argument: the total size of the strongly connected components is $n$.
The issue we are facing in taking advantage of this observation is that the sequence of sizes is not known a priori, it depends on the graph. 
To address this we use a universal sequence. 
First, here a sequence means a finite sequence of non-negative integer, for instance $(5,2,3,3)$.
The size of a sequence is the total sum of its elements, so $(5,2,3,3)$ has size $13$.
We say that $v = (v_1,\dots,v_k)$ embeds into $u = (u_1,\dots,u_{k'})$ if
there exists an increasing function $f : [1,k] \to [1,k']$ such that for all $i \in [1,k]$,
we have $v_i \le u_{f(i)}$.
For example $(5,2,3,3)$ embeds into $(4,6,1,2,4,1,3)$ but not in $(3,2,5,3,3)$.
A sequence $u$ is $n$-universal if all sequences of size at most $n$ embed into $u$.

\begin{rem}
The notion of $n$-universal sequence is equivalent to that of $(n,2)$-universal trees:
a tree of height $2$ is given by a sequence consisting of the number of leaves of each subtree of height $1$.
This explains why the constructions that follow may feel familiar after reading Section~\ref{sec:parity}.
\end{rem}

Let us define an $n$-universal sequence $u_n$, inductively on $n \in \N$.
We set $u_0 = ()$ (the empty sequence), $u_1 = (1)$, and 
$u_n$ is the concatenation of $u_{\lfloor n/2 \rfloor}$ with the singleton sequence $(n)$
followed by $u_{n - 1 - \lfloor n/2 \rfloor}$. 
Writing $+$ for concatenation, the definition reads $u_n = u_{\lfloor n/2 \rfloor} + (n) + u_{n - 1 - \lfloor n/2 \rfloor}$
Let us write the first sequences:
\[
u_2 = (1,2),\quad 
u_3 = (1,3,1),\quad
u_4 = (1,2,4,1),\quad
u_5 = (1,2,5,1,2),\quad
u_6 = (1,3,1,6,1,2), \dots
\]

\begin{lem}
The sequence $u_n$ is $n$-universal and has size $O(n \log(n))$.
\end{lem}
\begin{proof}
We proceed by induction on $n$. The case $n = 0$ is clear, let us assume that $n > 0$.
Let $v = (v_1,\dots,v_k)$ be a sequence of size $n$, we show that $v$ embeds into $u_n$.
There exists a unique $p \in [1,k]$ such that 
$(v_1,\dots,v_{p-1})$ has size smaller than or equal to $\lfloor n/2 \rfloor$
and $(v_1,\dots,v_p)$ has size larger than $\lfloor n/2 \rfloor$.
This implies that $(v_{p+1},\dots,v_k)$ has size at most $n - 1 - \lfloor n/2 \rfloor$.
By induction hypothesis $(v_1,\dots,v_{p-1})$ embeds into $u_{\lfloor n/2 \rfloor}$
and $(v_{p+1},\dots,v_k)$ embeds into $u_{n - 1 - \lfloor n/2 \rfloor}$,
so $v$ embeds into $u_n$.

The recurrence on size is $|u_n| = |u_{\lfloor n/2 \rfloor}| + n + |u_{n - 1 - \lfloor n/2 \rfloor}|$,
it implies that $|u_n|$ is bounded by $O(n \log(n))$.
\end{proof}

We now use the universal sequence to improve on Lemma~\ref{lem:strongly_connected_to_general_naive}

\begin{lem}
\label{lem:strongly_connected_to_general}
Let $\Omega$ be a positionally determined prefix independent objective and $n \in \N$.
For each $k \in [1,n]$, let $\U_k$ be a $(k,\Omega)$-universal graph for strongly connected graphs.
Let us write $u_n = (x_1,\dots,x_k)$, 
then $\U(u_n) = \langle \U_{x_1},\dots,\U_{x_k} \rangle$ is an $(n,\Omega)$-universal graph.
\end{lem}
\begin{proof}
We follow the same lines as for Lemma~\ref{lem:strongly_connected_to_general_naive}, 
in particular the same argument implies that $\U(u_n)$ satisfies $\Omega$.

Let $G=(V,E)$ be a graph satisfying $\Omega$, we show that $G$ maps into $\U(u_n)$.
We decompose $G$ into strongly connected components $G_1=(V_1,E_1),\dots,G_p=(V_p,E_p)$ as before. 
Let us write $v = (|V_1|,\dots,|V_p|)$ the sequence of sizes of the components.
The sequence $v$ has size at most $n$, implying that $v$ embeds into $u_n$:
there exists an increasing function $f : [1,p] \to [1,|u_n|]$ such that
for all $i \in [1,p]$ we have $|V_i| \le u_{f(i)}$.
It follows that for each $i \in [1,p]$, there exists a homomorphism $\phi_i : G_i \to \U_{f(i)}$.

Let us define $\phi : G \to \U(u_n)$ by $\phi(v) = \phi_i(v)$ if $v \in V_i$.
We now argue that $\phi$ is a homomorphism $\phi : G \to \U^n$.
Let $e = (v,c,v') \in E$, 
then either $e \in E_i$, and then $\phi(e)$ is an edge in $\U(u_n)$ 
because $\phi_i$ is a homomorphism,
or $e$ goes from $G_i$ to $G_j$ with $i < j$, and then the same holds
because $f$ is increasing.
Thus $\phi$ is a homomorphism and $\U(u_n)$ is $(n,\Omega)$-universal.
\end{proof}

\begin{rem}\label{rmk:complexity_of_universal_sequences}
To appreciate the improvement of Lemma~\ref{lem:strongly_connected_to_general} over Lemma~\ref{lem:strongly_connected_to_general_naive}, 
let us consider the case where the $(k,\Omega)$-universal graph $\U_k$ for strongly connected graphs has size $\alpha k$, where $\alpha$ does not depend on $k$.
Then Lemma~\ref{lem:strongly_connected_to_general_naive} yields an $(n,\Omega)$-universal graph
of size $\alpha n^2$ while Lemma~\ref{lem:strongly_connected_to_general} brings it down to $O(\alpha n \log(n))$.

\end{rem}

\subsection{Universal graphs for disjunctions of mean payoff objectives} 
\label{subsec:universal_graphs_disjunctions_mean_payoff}
We state in the following theorem an upper bound on the construction of universal graphs for disjunctions of mean payoff objectives.

\begin{thm}
\label{thm:upper_bound_universal_graph_disj_mean_payoff}
Let $n,d,N \in \N$ and $W \subseteq \Z$ of largest weight in absolute value $N$.
There exists an $(n,\bigvee_{i \in [1,d]} \MP{W}_i)$-universal graph of size $O(n \log(n) \cdot d N)$.
\end{thm}

As suggested by the previous subsection, we first construct $(n,\bigvee_{i \in [1,d]} \MP{W}_i)$-universal graphs
for strongly connected components.
The following result shows a decomposition property.

\begin{lem}
Let $G$ be a strongly connected graph and $W \subseteq \Z$.
Assume that $G$ satisfies $\bigvee_{i \in [1,d]} \MP{W}_i$ then there exists $i \in [1,d]$ such that $G$ satisfies $\MP{W}_i$.
\end{lem}

\begin{proof}
We prove the contrapositive: assume that for all $i \in [1,d]$ the graph does not satisfy $\MP{W}_i$,
implying that for each $i \in [1,d]$ there exists a negative cycle $C_i$ around some vertex $v_i$.
By iterating each cycle an increasing number of times, we construct a path in $G$ which does not satisfy $\bigvee_i \MP{W}_i$.

To make this statement formal, we use the following property:
for any $i \in [1,d]$, any finite path $\pi$ can be extended to a finite path $\pi \pi'$ with weight $\leq -1$ on the $i$\textsuperscript{th} component.
This is achieved simply by first going to $v_i$ using strong connectedness, then iterating through cycle $C_i$ a sufficient number of times.

We then apply this process repeatedly and in a cyclic way over $i \in [1,d]$ to construct an infinite path
such that for each $i \in [1,d]$, infinitely many times the $i$-th component is less than $-1$.
This produces a path which does not satisfy $\bigvee_{i \in [1,d]} \MP{W}_i$, a contradiction.
\end{proof}

Note that the above proof relies on the $\liminf$ semantics that we assumed for defining $\MP{W}$.
No similar simple decomposition is known for the $\limsup$ semantics, which turns out to be much harder (it is in fact $\NP$-complete even when $N$ is fixed~\cite{VelnerC0HRR15}).

\begin{cor}
\label{cor:universal_graph_strongly_connected_graphs_disj_mean_payoff}
Let $n \in \N, W \subseteq \Z$, and $\U$ be an $(n,\MP{W})$-universal graph for strongly connected graphs.
We construct $d \cdot \U$ the disjoint union of $d$ copies of $\U$, where the $i$-th copy reads the $i$-th component.
Then $d \cdot \U$ is an $(n,\bigvee_{i \in [1,d]} \MP{W}_i)$-universal graph for strongly connected graphs. 
\end{cor}

We can now prove Theorem~\ref{thm:upper_bound_universal_graph_disj_mean_payoff}.

Thanks to Lemma~\ref{lem:upper_bounds_universal_graphs_mean_payoff}, there exists an $(n,\MP{W})$-universal graph of size $nN$.
Thanks to Corollary~\ref{cor:universal_graph_strongly_connected_graphs_disj_mean_payoff}, this implies
an $(n,\bigvee_{i \in [1,d]} \MP{W}_i)$-universal graph for strongly connected graphs of size $ndN$.
Now Lemma~\ref{lem:strongly_connected_to_general} yields an $(n,\bigvee_i \MP{W}_i)$-universal graph 
of size $O(n \log(n) \cdot d N)$ (see Remark~\ref{rmk:complexity_of_universal_sequences}).

\subsection{The complexity of solving disjunctions of mean payoff games using universal graphs}
\label{subsec:complexity_disjunctions_mean_payoff}

Let us note here that there are actually two variants of the mean payoff objective: using infimum limit or supremum limit.
In many cases the two variants are equivalent: 
this is the case when considering mean payoff games or disjunctions of parity and mean payoff games.
However, this is not true anymore for disjunctions of mean payoff objectives, as explained in~\cite{VelnerC0HRR15}.
Our constructions and results apply using the infimum limit, and do not extend to the supremum limit.

The main result related to disjunction of mean payoff games in~\cite{VelnerC0HRR15} (Theorem~6) is that the problem is in $\NP \cap \coNP$ and can be solved in time $O(m \cdot n^2 \cdot d \cdot N)$.
(Note that since we consider the dual objectives, the infimum limit becomes a supremum limit in~\cite{VelnerC0HRR15}.)

Combining Theorem~\ref{thm:upper_bound_universal_graph_disj_mean_payoff} with Theorem~\ref{thm:solving_product_game_linear} yields the following result.

\begin{cor}
Let $n,m,d,N \in \N$.
In the RAM model with word size $w = \log(n) + \log(N) + \log(d)$, 
there exists an algorithm for solving disjunctions of mean payoff games with weights in $(-N,N)$ 
with time complexity $O(m \cdot n \log(n) \cdot d \cdot N)$ and space complexity $O(n)$.
\end{cor}

Note that for the choice of word size, the two tasks for manipulating universal graphs, 
namely computing $\delta(q,w)$ and checking whether $q \le q'$ 
are indeed unitary operations as they manipulate numbers of order $nN$.

\section{Conclusions}
\label{sec:conclusions}
In this paper we have introduced the notion of universal graphs for constructing generic value iteration algorithms
for games with positionally determined objectives.
We have instantiated this to four classes of games, achieving or improving the state of the art for solving these games.
There are many other potential applications, for instance Rabin objectives or disjunctions of parity objectives.

We laid the first brick for the theory of universal graphs, and we hope that this new tool will open new perspectives.
An exciting research direction is to construct universal graphs for subclasses of graphs. 
For instance, it is known that solving parity games over graphs of bounded tree width or clique width is polynomial~\cite{Obdrzalek03,Obdrzalek07}. Are there universal graphs of polynomial size for these classes of graphs?
A step in this direction was made recently by Daviaud, Jurdzi{\'n}ski, and Thejaswini~\cite{DaviaudJT20}: revisiting the register games algorithm~\cite{Leh18,LehtinenB20}, they introduced a new parameter called the Strahler number and constructed universal trees for bounded Strahler numbers.

Another promising direction is to extend the use of universal graphs from games to automata. 
The first brick in the wall was thrown by Daviaud, Jurdzi{\'n}ski, and Lehtinen~\cite{DaviaudJL19},
who showed that universal trees can be used to improve constructions for alternating automata over infinite words.

\bibliographystyle{alphaurl}
\bibliography{bib}

\appendix
\section{Analysis of the function $F$}
\label{sec:analysis_F}

Define $F(p,h) = f(2^p,h)$ for $p \ge 0$ and $h \ge 1$.
Then we have
\[\begin{array}{lll}
F(p,h) & \le & F(p,h-1) + 2 F(p-1,h), \\
F(p,1) & = & 2^p, \\
F(0,h) & = & 1.
\end{array}\]
To obtain an upper bound on $F$ we define $\overline{F}$ by
\[\begin{array}{lll}
\overline{F}(p,h) & = & \overline{F}(p,h-1) + 2 \overline{F}(p-1,h), \\
\overline{F}(p,1) & = & 2^p, \\
\overline{F}(0,h) & = & 1,
\end{array}\]
so that $F(p,h) \le \overline{F}(p,h)$.
Define the bivariate generating function 
\[
\F(x,y) = \sum_{p \ge 0, h \ge 1} \overline{F}(p,h) x^p y^h.
\]
Plugging the inductive equalities we obtain
\[
\F(x,y) = \frac{y}{1 - 2x - y},
\]
from which we extract that $\overline{F}(p,h) = 2^p \binom{p+h-1}{p}$, implying $F(p,h) \le 2^p \binom{p+h-1}{p}$.
Putting everything together we obtain
\[
f(n,h) \le 2^{\lceil \log(n) \rceil} \cdot \binom{\lceil \log(n) \rceil + h - 1}{\lceil \log(n) \rceil}.
\]

\section{Analysis of the function $G$}
\label{sec:analysis_G}

Define $G(p,h) = g(2^p,h)$ for $p \ge 0$ and $h \ge 1$.
Then we have
\[\begin{array}{lll}
G(p,h) & \ge & \sum_{k = 0}^p G(p-k,h-1), \\
G(p,1) & \ge & 1, \\
G(0,h) & = & 1.
\end{array}\]
To obtain a lower bound on $G$ we proceed similarly as for $F$.
We define $\overline{G}$ by
\[\begin{array}{lll}
\overline{G}(p,h) & = & \overline{G}(p,h-1) + \overline{G}(p-1,h), \\
\overline{G}(p,1) & = & 1, \\
\overline{G}(0,h) & = & 1,
\end{array}\]
so that $G(p,h) \ge \overline{G}(p,h)$.
We verify by induction that $\overline{G}(p,h) = \binom{p+h-1}{p}$, which follows from the identity
\[
\binom{p+h-1}{p} = \binom{p+h-2}{p} + \binom{p+h-2}{p-1}.
\]
This implies that $G(p,h) \ge \binom{p+h-1}{p}$.
Putting everything together we obtain
\[
g(n,h) \ge \binom{\lfloor \log(n) \rfloor + h - 1}{\lfloor \log(n) \rfloor}.
\]

\newpage 

\section{Pseudo-code for linear time algorithm solving safety games}\label{sec:safety}

\begin{algorithm*}
 \KwData{A safety game $\game=(V,E)$}
 \SetKwFunction{FUpd}{Upd}
 \SetKwFunction{FSolve}{Solve}
 \SetKwProg{Fn}{Function}{:}{}
 \DontPrintSemicolon
 
\Fn{\FSolve{$\game$}}{
	$X \gets V$ \;
	$R \gets \{v \in \VE \mid v \text{ is a sink}\}$ \;
	\For{$v \in \VE$}{	
		$\cnt(v) \gets $ number of edges outgoing from $v$
	}
	\vskip1em
	\While{$R \neq \emptyset$}{
		$R \gets \FUpd(X,R,\cnt)$
	}
	\Return{$X$}
}

\vskip1em
\Fn{\FUpd{$X,R,\cnt$}}{
	$R' \gets \emptyset$ \;
	\For{edge $e=(v,v')$ from $X$ to $R$}{
		\If{$v \in \VA$}{
			$R' \gets R' \cup \{v\}$
		}
	}
	\Else{
		$\cnt(v) = \cnt(v) -1$ \;
		\If{$\cnt(v) = 0$}{
			$R' \gets R' \cup \{v\}$
		}
	}
	$X \gets X \setminus R'$ \;
	\Return{$R'$}
}
\caption{A linear time algorithm for solving safety games. Note that the procedure $\texttt{Upd}$, which computes $X_{i+1}$ and $R_i$ from $X_i$ and $R_{i-1}$, updates $X$ and $\cnt$ in place.}
\label{algo:value_iteration}
\end{algorithm*}

\newpage

\section*{Pseudo-code for generic value iteration algorithm}
\label{sec:product_safety}

\begin{algorithm*}[h]
 \KwData{A game $\game=(V,E)$ and a linear graph $\auto=(Q,\Delta)$.}
 \SetKwFunction{FUpd}{Upd}
 \SetKwFunction{FSolve}{Solve}
 \SetKwProg{Fn}{Function}{:}{}
 \DontPrintSemicolon
 
\Fn{\FSolve{$\game, \auto$}}{
	\For{$v \in V$}{
        $\theta(v) \leftarrow \min Q$
	}

	\For{$v \in \VE$}{
        $\cnt(v) \leftarrow |\{(v,c,v') \in E \mid \theta(v') \leq \delta(\theta(v),c)\}|$
	}

	$\Inv \leftarrow \{v \in \VE \mid \cnt(v)=0 \} \cup \{v \in \VA \mid \exists (v,c,v') \in E, \theta(v') > \delta(\theta(v),c)\}$ \;
	\vskip1em
	
	\While{$\Inv \neq \emptyset$}{
		$\Inv \leftarrow \FUpd(\theta, \Inv, \cnt)$
	}
	\Return{$\{v \in V \mid \theta(v) < \top\}$}
}

\vskip1em
\Fn{\FUpd{$\theta, \Inv, \cnt$}}{
	$\Inv' \gets \emptyset$ \;
	\For{$v \in \Inv$}{
		\If{$v \in \VE$}{$\theta'(v) \gets \min\{\rho(\theta(v'),c), (v,c,v') \in E\}$}
		\Else{$\theta'(v) \gets \max\{\rho(\theta(v'),c), (v,c,v') \in E\}$}
		\For{$e=(u,c,v) \in E$ such that $\theta(v) \leq \delta(\theta(u),c) < \theta'(v)$}{
			\If{$u \in \VA$}{
				$\Inv' \gets \Inv' \cup \{u\}$
			}
			\Else{
				$\cnt(u) \gets \cnt(u) - 1$ \;
				\If{$\cnt(u) = 0$}{
					$\Inv' \gets \Inv' \cup \{u\}$
				}
			}
		}
		$\theta(v) \gets \theta'(v)$ \;
		\If{$v \in \VE$}{
			$\cnt(v) \gets |\{(v,c,v') \in E \mid \theta(v') \leq \delta(\theta(v),c)\}|$
		}
	}
	\Return{$\Inv'$}
}
\caption{A generic value iteration algorithm obtained by solving the safety game $\game \chain \auto$ when $\auto$ is linear.}
\label{alg:product_safety_game}
\end{algorithm*}

\end{document}